\documentclass[pra,aps,superscriptaddress,twocolumn,showpacs,nofootinbib]
{revtex4-1}
\usepackage{placeins}
\usepackage{amsfonts, amsmath,amssymb,amsthm,bm}
\usepackage{mathrsfs}
\usepackage{graphicx}
\usepackage[usenames,dvipsnames]{color}

\usepackage{framed}

\usepackage{ulem}

\usepackage{palatino}
\usepackage{fancyhdr}
\usepackage{mathrsfs}
\usepackage{multirow}

\usepackage{braket}
\usepackage[pdftex,
            pdfauthor={Marius Krumm, Howard Barnum, Jonathan Barrett, Markus P. Mueller},
            pdftitle={Thermodynamics and the structure of quantum theory},
            pdfsubject={Quantum Foundations, Quantum Information Theory, Quantum Thermodynamics, Structure},
            pdfkeywords={Quantum, Structure, Generalized Probabilistic Theories, Second Law, GPT, Observables, Foundations, Physics, Information Theory, Neumann, Entropy, membrane, semi-permeable, Thermodynamics, Postulates},
            pdfproducer={Latex},
            pdfcreator={pdflatex}]{hyperref}

\usepackage{etoolbox}

\hypersetup{
   colorlinks=true,
   linkcolor=blue,
   citecolor=blue,
   urlcolor=Violet
}

\renewcommand*{\thefootnote}{\fnsymbol{footnote}}

\newcommand{\join}{\vee}

\newcommand{\R}{\mathbb{R}}
\newcommand{\be}{\begin{equation}}
\newcommand{\ee}{\end{equation}}

\newcommand{\tr}{{\rm tr}}

\theoremstyle{plain}
\newtheorem{thm}{Theorem}

\newtheorem{lem}[thm]{Lemma}
\newtheorem{assumptions}[thm]{Assumptions}

\newtheorem{defn}[thm]{Definition}
\newtheorem{definition}[thm]{Definition}
\newtheorem{observation}[thm]{Observation}
\theoremstyle{remark}

\makeatletter
\def\p@subsection{}\makeatother

\begin{document}
\title{Thermodynamics and the structure of quantum theory}
\author{Marius Krumm}
\affiliation{Department of Applied Mathematics, University of Western Ontario, London, ON N6A 5BY, Canada}
\affiliation{Department of Theoretical Physics, University of Heidelberg, Heidelberg, Germany}
\affiliation{Perimeter Institute for Theoretical Physics, Waterloo, ON N2L 2Y5, Canada}
\author{Howard Barnum}
\affiliation{Department of Physics and Astronomy, University of New Mexico, Albuquerque, NM, USA}
\affiliation{Institute for Theoretical Physics and Riemann Center for Geometry 
and Physics, Leibniz University Hannover, Appelstra\ss e 2, 30163 Hannover, 
Germany}
\author{Jonathan Barrett}
\affiliation{Department of Computer Science, University of Oxford, Oxford, UK}
\author{Markus P. M\"uller}
\affiliation{Department of Applied Mathematics, University of Western Ontario, London, ON N6A 5BY, Canada}
\affiliation{Department of Philosophy, University of Western Ontario, London, ON N6A 5BY, Canada}
\affiliation{Perimeter Institute for Theoretical Physics, Waterloo, ON N2L 2Y5, Canada}
\affiliation{Department of Theoretical Physics, University of Heidelberg, Heidelberg, Germany}

\begin{abstract}
Despite its enormous empirical success, the formalism of quantum theory still raises fundamental questions: why is nature described in terms of complex Hilbert spaces, and what modifications of it could we reasonably expect to find in some regimes of physics? Here we address these questions by studying how compatibility with thermodynamics constrains the structure of quantum theory. We employ two postulates that any probabilistic theory with reasonable thermodynamic behavior should arguably satisfy. In the framework of generalized probabilistic theories, we show that these postulates already imply important aspects of quantum theory, like self-duality and analogues of projective measurements, subspaces and eigenvalues. However, they may still admit a class of theories beyond quantum mechanics. Using a thought experiment by von Neumann, we show that these theories admit a consistent thermodynamic notion of entropy, and prove that the second law holds for projective measurements and mixing procedures. Furthermore, we study additional entropy-like quantities based on measurement probabilities and convex decomposition probabilities, and uncover a relation between one of these quantities and Sorkin's notion of higher-order interference.
\end{abstract}

\date{April 5, 2017}

\maketitle

\normalem

\bigskip

\renewcommand*{\thefootnote}{\arabic{footnote}}

\section{Introduction}

 Quantum mechanics has existed for about 100
years now, but despite its enormous
 success in experiment and
application, the meaning and origin of its
 counterintuitive
formalism is still widely considered to be difficult
 to grasp. Many attempts to put quantum mechanics on a more intuitive footing have been made over the decades, which includes the development of a variety of interpretations of quantum physics (such as the many-worlds interpretation~\cite{Everett}, Bohmian mechanics~\cite{Bohm}, QBism~\cite{QBism}, and many others~\cite{Madness}), and a thorough analysis of its departure from classical physics (as in Bell's Theorem~\cite{Bell} or in careful definitions of notions of contextuality~\cite{Spekkens}).
In more recent
 years, researchers, mostly coming
  from and inspired by the field of quantum information processing
  (early examples include~\cite{Hardy, BarnumLogic, BarnumConvex}), 
have taken as a starting point the set of all probabilistic theories.
Quantum theory is one of them and can be uniquely determined by specifying
 some of its characteristic properties~\cite{Fuchs} (as in e.g.~\cite{Hardy,DakicBrukner,MasanesMueller,InformationalDerivation,Hardy2011,Postulates,WilceDeriv1,MMAP2013,Hoehn,HoehnWever}).

While the origins of this framework date back at
  least to the 1960s~\cite{Mackey,MielnikGeo,LudwigArticles}, it was the development
  of quantum information theory with its emphasis on simple
  operational setups that led to a new wave of interest in
  ``generalized probabilistic theories'' (GPTs)~\cite{Hardy,Barrett}.
This framework turned out to be very fruitful for fundamental
investigations of quantum theory's information-theoretic and operational properties. For example, GPTs make it possible to contrast quantum
  information theory with other possible theories of information
  processing, and in this way to gain a deeper understanding of its
  characteristic properties in terms of computation or communication.

In a complementary approach, there has been a wave of attempts to
find simple physical principles that single out quantum correlations
from the set of all non-signalling correlations in the
device-independent formalism~\cite{Popescu}. These include non-trivial
communication complexity~\cite{vanDam}, macroscopic
locality~\cite{ML}, or information
causality~\cite{Pawlowski}. However, none of these principles so far
turns out to yield the set of quantum correlations exactly. This led
to the discovery of ``almost quantum
correlations''~\cite{AlmostQuantum} which are more general than those
allowed by quantum theory, but satisfy all the aforementioned
principles. Almost quantum correlations seem
to appear naturally in the context of quantum gravity~\cite{Craig}.
 
A relation to other fields of physics can also be drawn from
information causality, which can be understood as the requirement
that
 a notion of
entropy~\cite{WehnerEntropy,BarnumEntropy,KimuraEntropy,LogicEntropy}
exists which has some natural properties like the
 data-processing
inequality~\cite{Dahlsten}. These emergent connections to entropy
and quantum gravity are particularly interesting since they
 point to
an area of physics where modifications of quantum theory are
well-motivated: Jacobson's results~\cite{Jacobson} and holographic
duality~\cite{Ryu} relate thermodynamics, entanglement, and (quantum)
gravity, and modifying quantum theory has been discussed as a means
to
 overcome apparent paradoxes in black-hole
physics~\cite{HorowitzMaldacena}.

While generalized probabilistic theories provide a way to generalize
quantum theory and to study more general correlations and physical
theories, they still leave open the question as to which principles
should guide us in applying the GPT formalism
  for this purpose. The considerations above
 suggest taking, as a
guideline for such modifications, the principle that they 
 support a well-behaved notion of
thermodynamics. As A.\ Einstein \cite{Einstein} put it,
 
\textit{``A theory is the more impressive the greater the simplicity
  of its premises, the more different kinds of things it relates, and
  the more extended its area of applicability. Therefore the deep
  impression that classical thermodynamics made upon me. It is the
  only physical theory of universal content which I am convinced will
  never be overthrown, within the framework of applicability of its
  basic concepts.''}
 
 Along similar lines, A.\ Eddington
\cite{Eddington} argued that \textit{``The law that
 entropy always
  increases holds, I think, the supreme position among
 the laws of
  Nature. If someone points out to you that your pet
 theory of the
  universe is in disagreement with Maxwell's equations —
 then so
  much the worse for Maxwell's equations. If it is found to be
  contradicted by observation — well, these experimentalists do
  bungle
 things sometimes. But if your theory is found to be against
  the
 second law of thermodynamics I can give you no hope; there is
  nothing for it but to collapse in deepest humiliation.''}
 
Here we take this point of view
 seriously. We investigate
what kinds of probabilistic theories,
 including but not limited to
quantum theory, could peacefully coexist
 with
thermodynamics. We present two postulates that formalize 
  important physical properties which can be expected 
 to hold in any such theory. On the one hand, these two postulates allow for a
class of theories
 more general than quantum or classical theory,
which thus describes potential
 alternative physics consistent with
important parts of thermodynamics as 
 we know it. Indeed, by considering a thought experiment originally conceived by
von Neumann, we show that these theories all give rise to a unique, consistent
 form of thermodynamical entropy. Furthermore, we show that this entropy satisfies several other important properties, including two instances of the
 second law. On the other hand, we show that these
postulates already 
 imply many structural properties which are also
present in quantum
 theory, for example self-duality and the
existence of
 analogues of projective measurements, observables,
eigenvalues and eigenspaces.
 
 In summary, our analysis shows that
important structural aspects of
 quantum and classical theory are
already implied by these aspects of thermodynamics,
 but on the other
hand it suggests that there is still some ``elbow room'' for
 modification
within these limits dictated by thermodynamics.

Thermodynamics in GPTs has been considered in some earlier
works. In~\cite{GPTBH1,GPTBH2}, the authors introduced a notion of
(R\'enyi-2-)entanglement entropy, and studied the phenomenon of
thermalization by
entanglement~\cite{PopescuShortWinter,Goldstein,Adlam} and the
black-hole information problem (in particular the Page
curve~\cite{Page}) in generalizations of quantum theory. H\"anggi and
Wehner~\cite{HaenggiWehner} have related the uncertainty principle to
the second law in the framework of GPTs. Chiribella and Scandolo
(\cite{diagonal,ChiribellaScandolo}, see also~\cite{ScandoloMaster})
have considered the notion of diagonalization and majorization in
general theories, leading to a resource-theoretic approach to
thermodynamics in GPTs. There are various
  connections between their results and ours, but there are essential
  differences. In particular, they assume the purification postulate
  (which is arguably a strong assumption that in particular excludes
  classical thermodynamics), whereas we are not making any assumption
  on composition of systems whatsoever, and in this sense work in a
  more general framework. Furthermore, while Chiribella and Scandolo
  take a resource-theoretic approach motivated by quantum information
  theory, our analysis relies on a more traditional thermodynamical
  thought experiment (namely von Neumann's). We presented 
  results related to some of those in the present paper in the conference
  proceedings~\cite{BBKM}; here we use different
  assumptions and obtain additional results.
 
 Our paper is organized as follows. We start with an overview of the
 framework of generalized probabilistic theories. Then we present von
 Neumann's thought experiment on thermodynamic entropy, and a
 modification of it due to Petz~\cite{Petz}. Although it relies on
 very mild assumptions, it already rules out all theories that admit a
 state space known as the \textit{gbit} or \textit{squit} (a
 square-shaped state space that can be used to describe one of the two
 local subsytems of a composite system known as the
 PR-box~\cite{PopescuRohrlich}, exhibiting stronger-than-quantum
 correlations).  Then we present our two postulates, and show that
 they imply many structural features of quantum theory. We show that
 theories that satisfy both postulates behave consistently in von
 Neumann's thought experiment and admit a notion of thermodynamic
 entropy which satisfies versions of the second law.

Because entropies
are an important bridge between information theory and thermodynamics,
in the
 final section we investigate the consequences of our postulates for generalizations of quantities of known
  significance in quantum thermodynamics~\cite{SecondLaws}, defined by applying R\'enyi entropies to probabilities in
  convex decompositions of a state, or of measurements made on a
  state. In particular, we show a relation between max-entropy 
and Sorkin's notion of higher-order interference~\cite{Sorkin}:   
equality of the preparation and measurement 
based max-entropies implies the absence of 
higher-order interference.
Most proofs are deferred to the
  appendix. Several results of this paper have been announced in the Master thesis of one of the authors~\cite{MariusMaster}.

\section{The mathematical framework}
\label{SecFramework}
Our results are obtained in the framework of generalized probabilistic
theories (GPTs)~\cite{Mielnik,Barrett,Hardy,MasanesMueller,BBLW}. The
goal of this framework is to capture all probabilistic theories,
i.e.\ all theories that use states to make predictions for
probabilities of measurement outcomes. Although the framework is based
on very weak and natural assumptions, we can only provide a short
introduction of the main notions and results here. For more detailed
explanations of the framework, see
e.g.\ ~\cite{MariusMaster,Barrett,MasanesMueller,Hardy,Janotta,PfisterNat}. The
framework contains quantum theory and also the application of
probability theory to classical physics, often referred to as
classical probability theory, as special cases. It also contains theories which
  differ substantially from classical or quantum probability theory,
  for example boxworld~\cite{Barrett}, which allows superstrong
  nonlocality, and theories that allow higher-order interference
  ~\cite{Sorkin}.
 
 A central notion is that of the state and the set of states, the
 state space $\Omega_A$. A state contains all information necessary to
 calculate all probabilities for all outcomes of all possible
 measurements. One possible and convenient representation would be to
 simply list the probabilities of a set of ``fiducial'' measurement
 outcomes which is sufficient to calculate all outcome probabilities
 for all measurements~\cite{Barrett,Hardy}. An example is given in
 Figure~\ref{fig_gbitcone}.
 
\begin{figure}
\begin{center}
\includegraphics[width=0.4\textwidth]{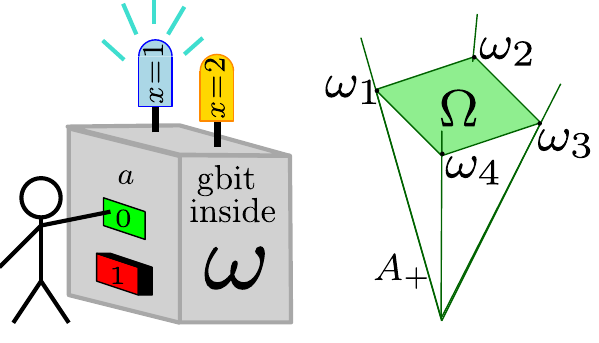}
\end{center}
\caption{\small \textit{An example state space, $A$,
    modelling a so-called ``gbit''~\cite{Barrett} which is often used
    to describe one half of a PR-box. The operational setup is
    depicted on the left, and the mathematical formulation is sketched
    on the right. An agent (``Alice'') holds a black box $\omega$ into
    which she can input one bit, $a\in\{0,1\}$, and obtains one
    output, $x\in\{1,2\}$. The box is described by a conditional
    probability $p(x|a)$. In the GPT framework, $\omega$ becomes an
    actual state, i.e.\ an element of some state space
    $\Omega$. Concretely, $\omega=(1,p(1|0),p(1|1))\in\mathbb{R}^3$,
    where the first entry $1$ is used to describe the normalization,
    $p(1|0)+p(2|0)=p(1|1)+p(2|1)$. In this case, all probabilities are
    allowed by definition, so that the state space $\Omega$ becomes
    the square, i.e.\ the points $(1,s,t)$ with $0\leq s,t \leq
    1$. Alice's input $a$ is interpreted as a ``choice of
    measurement'', and the two measurements are $e_{x=1}^{(a=0)},
    e_{x=2}^{(a=0)}$ resp.\ $e_{x=1}^{(a=1)}, e_{x=2}^{(a=1)}$ such
    that $\sum_{x=1}^2 e_x^{(a)}(\omega)=1$ for all states
    $\omega\in\Omega$. If we describe effects by vectors by using the
    standard inner product, we have, for example,
    $e_{x=1}^{(a=0)}=(0,1,0)$, since
    $e_{x=1}^{(a=0)}(\omega)=P(1|0)=(0,1,0)\cdot\omega$. There are
    four pure states, labelled $\omega_1,\ldots,\omega_4$. Every pure
    state $\omega_i$ is perfectly distinguishable from every other
    pure state $\omega_j$ for $j\neq i$, but no more than two of them
    are jointly distinguishable in a single measurement. More
    generally, every state on one side of the square is perfectly
    distinguishable from every state on the opposite side. The unit
    effect is $u_A=(1,0,0)$.}}\label{fig_gbitcone}
\end{figure}
 
 It is possible to create statistical mixtures of states: Let us
 assume a black box device randomly prepares a state $\omega_1$ with
 probability $p_1$ and a state $\omega_2$ with probability $p_2$. In
 agreement with the representation of states as lists of probabilities
 and the law of total probability, the appropriate state to describe
 the resulting measurement statistics is $\omega = p_1 \omega_1 + p_2
 \omega_2$. This means that the state space $\Omega_A$ is convex and
 is embedded into a real vector space $A$ (to be described below). Due
 to the interpretation of states as lists of probabilities (which are
 between $0$ and $1$) we demand that $\Omega_A$ is bounded. Any state
 that cannot be written as a convex decomposition of other states is
 called a pure state. As pure states cannot be interpreted as
 statistical mixtures of other states, they are also called states of
 maximal knowledge. Furthermore, there is no physical distinction
 between states that can be prepared exactly, and states that can be
 prepared to arbitrary accuracy.  Thus, we also assume that $\Omega_A$
 is topologically closed.  In order to not obscure the physics by the
 mathematical technicalities introduced by infinite dimensions, we
 will assume that $A$ is finite-dimensional. Thus $\Omega_A$ is
 compact.  Consequently, every state can be obtained as a statistical mixture of finitely many pure states~\cite{Webster}.

Furthermore, it turns out to be convenient to introduce
  unnormalized states $\omega$, defined as the non-negative multiples
  of normalized states. They form a closed convex cone $A_+ := \mathbb
  R_{\ge 0} \cdot \Omega_A$. For simplicity of description, we choose
  the vector space containing the cone of states to be of minimal
  dimension, i.e.\ ${\rm span}(A_+)=A$.
 
 We introduce the
  normalization functional $u_A : A \rightarrow \mathbb R$ which
  attains the value one on all normalized states,
  i.e.\ $u_A(\omega)=1$ for all $\omega\in\Omega_A$.
 It is linear,
non-negative on the whole cone, zero only for the origin, and
$\omega\in A_+$ is an element of $\Omega_A$ if and only if
$u_A(\omega)=1$. The normalization $u_A(\omega)$ can be interpreted as
the probability of success of the preparation procedure. For states
with $u_A(\omega) < 1$, the preparation succeeds with probability
$u_A(\omega)$. The states with normalization $>1$ do not have a
physical interpretation, but adding them allows us to take full
advantage of the notion of cones from convex geometry.
 
 Effects are functionals
that map (sub)normalized states to probabilities, i.e.\ into
$[0,1]$. To each measurement outcome we assign an effect that
calculates the outcome probability for any state. Effects have to be
linear  for consistency with
the statistical mixture interpretation of
convex combinations of states. A \emph{measurement} (with $n$ outcomes) is a collection of effects $e_1,\ldots,e_n$ such that $e_1+\ldots+e_n=u_A$. Its interpretation is that performing the measurement on some state $\omega\in\Omega_A$ yields outcome $i$ with probability $e_i(\omega)$.

A set of states $\omega_1,\ldots,\omega_n$ is called \emph{perfectly distinguishable} if there exists a measurement $e_1,\ldots,e_n$ such that $e_i(\omega_j)=\delta_{ij}$, that is, $1$ if $i=j$ and $0$ otherwise. A collection of $n$ perfectly distinguishable pure states is called an \emph{$n$-frame}, and a frame is called \emph{maximal} if it has the maximal number $n$ of elements possible in the given state space. In quantum theory, for example, the maximal frames are exactly the orthonormal bases of Hilbert space. In more detail, a frame on an $N$-dimensional quantum system is given by $\omega_1=|\psi_1\rangle\langle\psi_1|,\ldots,|\psi_N\rangle\langle\psi_N|$, where $|\psi_1\rangle,\ldots,|\psi_N\rangle$ are orthonormal basis vectors.

Transformations are maps $T:A\to A$ that map states to states, i.e.\ $T(A_+)\subseteq A_+$. Similarly as effects, they also have to be linear in order to preserve statistical mixtures. They cannot increase the total probability, but are allowed to decrease it (as is the case, for example, for a filter), thus $u_A \circ T(\omega) \le u_A(\omega)$ for all $\omega \in
A_+$.
 
 Instruments\footnote{Some authors have recently 
begun referring to instruments as \emph{operations}, but long-standing
convention in  quantum information
theory (including~\cite{NielsenChuang}) uses the term ``operation'' 
for the quantum case of what we are calling transformations
(which are completely positive maps).  Also, Davies and Lewis~\cite{Davies}  define instrument more generally, to allow 
for continuously-indexed transformations, where we only consider finite collections $T_j$.}
~\cite{Davies} 
are collections of
transformations $T_j$ such that $\sum_j u_A \circ T_j = u_A$.
If an
instrument is applied to a state $\omega$, one obtains outcome $j$
(and post-measurement state $T_j(\omega)/p_j$) with probability
$p_j:=u_A(T_j(\omega))$. 
Each instrument corresponds to a measurement
given by the effects $u_A \circ T_j$. We will say 
it ``induces'' this measurement.

The framework
  of GPTs does not assume \emph{a priori} that all mathematically
  well-defined states, transformations and measurements can actually
  be physically implemented.
 Here, we will assume that a
measurement constructed from physically allowed effects is also
physically allowed. Moreover, we assume that the set of allowed
effects has the same dimension as $A_+$, because otherwise there would
be distinct states that could not be distinguished by any
measurement.

\section{Von Neumann's thought experiment}
\label{SecThoughtExp}
The following thought experiment has been applied by von
Neumann~\cite{Neumann} to find a notion of thermodynamic entropy for
quantum states $\rho$.  The result turns out to equal von Neumann
entropy, $H(\rho)=-\tr(\rho\log\rho)$. We apply the thought experiment
to a wider class of probabilistic theories.

We adopt the physical picture  used by von
Neumann~\cite{Neumann} to describe the thought experiment\footnote{Our thought experiment is identical to von Neumann's, up to two differences: first, we translate all quantum notions to more general GPT notions; second, while von Neumann implements the transition from (5) to (6) in Figure~\ref{FIG:NeumannTotal} via sequences of projections, we implement this transition directly via reversible transformations.}; we will comment on some idealizations used in this model at the end of this section. We consider a GPT ensemble $[S_1,...,S_N]$, where $S_i$ denotes the $i$-th physical system, and $N_j$ of the systems
are in state $\omega_j$, where $j=1,\ldots,n$ and $\sum_j N_j=N$. This ensemble is described by the state $\omega = \sum_{j=1}^n p_j
\omega_j$, where $p_j=N_j/N$, which describes the effective state of a system that is drawn uniformly at random from the ensemble.

 We introduce $N$ small, indistinguishable, hollow
 boxes\footnote{For a more detailed discussion of the physical properties of these small boxes, we refer the reader to von Neumann`s original work~\cite{Neumann}.}, and we put
each ensemble system $S_j$ into one of the boxes
 such that the
system is completely isolated from the outside.
 Furthermore, we
assume that the boxes form an ideal gas, which will
 allow us to use
the ideal gas laws in the following derivation. This
 gas will be called the $\omega$-gas. We
    will denote the total thermodynamic entropy of a system by $H$,
    with a subscript which may indicate whether it is the total
    entropy of a gas, which potentially depends both on the states of
    the GPT systems in the boxes and on the classical degrees of
    freedom (positions, momenta) of the boxes, or just the entropy
of the GPT or of the classical degrees of freedom.

At first we need to investigate how the entropy of the gas
  and of the
 ensemble are related to each other because later on, we
  will only consider the gas. So we consider also a second GPT
  ensemble $[S'_1,...,S'_N]$
 (described by $\omega' \in \Omega_A$)
  implanted into a gas the same way. At
 temperature $T=0$, the
  movement of the boxes freezes out and we are
 left with the GPT
  ensembles. In this case, the thermodynamic entropies of
    the gases  and the GPT ensembles must satisfy:
  $H_{\omega\text{-gas}}-H_{\omega'\text{-gas}} =
  H_{\omega\text{-ensemble}} -
  H_{\omega'\text{-ensemble}}$. Remember that the heat capacity is $C
  =
 \delta Q/\mathrm d T$, and as the gases only differ in their
  internal
 systems, which are isolated, $C$ is the same for both
  gases. With
 $\mathrm d H = \delta Q/T$ we thus find that
  $H_{\omega\text{-gas}}-H_{\omega'\text{-gas}}$ is constant in $T$,
  i.e. $H_{\omega\text{-gas}}-H_{\omega'\text{-gas}} =
  H_{\omega\text{-ensemble}} -
 H_{\omega'\text{-ensemble}}$ for all
  temperatures.
 
 The central tool for the thought experiment is a
  semi-permeable
 membrane. Whenever a box reaches the membrane, the
  membrane opens that
 box and measures the internal system. Depending
 on the result, a window is opened to let
  the box pass, or the window
 remains closed. It is crucial to note
  that this membrane will not cause problems in the style of Maxwell's
  demon, as was already discussed by von Neumann himself, because the
  membrane does not distinguish between its two sides.  
\begin{figure*}
\begin{center}
   \includegraphics[width=0.75\textwidth]{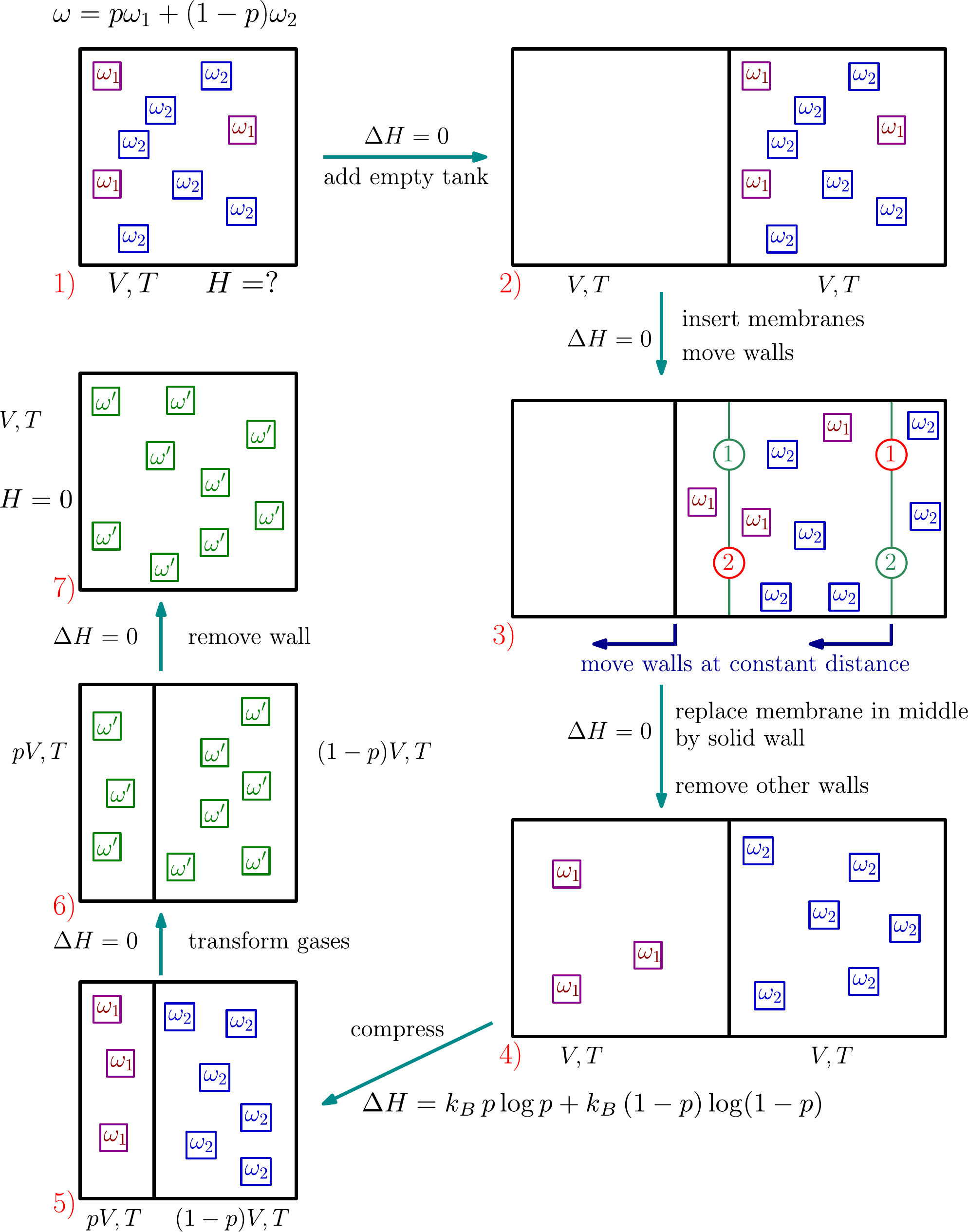}
	\caption{\small \textit{This figure shows von Neumann's thought experiment, as described in the main text. Stages 1)--5) also feature in Petz' version.}}\label{FIG:NeumannTotal}
\end{center}
\end{figure*}

Now we begin with the experiment itself; see
Figure~\ref{FIG:NeumannTotal}. We consider a state $\omega=
\sum_{j=1}^n p_j \omega_j$ where $\omega_j$ are perfectly
distinguishable pure states, and $p_j = N_j/N$, where $N_j$ boxes
contain a system in the state $\omega_j$. We assume that the
$\omega$-gas is confined in a container of volume $V$. Let there be a
second container which is identical to the first one, but empty. The
containers are merged together, the wall of the non-empty container
separating the containers replaced by a semi-permeable membrane which
lets only $\omega_1$ pass.  At the opposite wall of the non-empty
container we insert a semi-permeable membrane which only blocks
$\omega_1$. The solid wall in the middle and the outer semi-permeable
membrane are moved at constant distance until the solid wall hits the
other end.

Once this is accomplished, i.e.\ in stage 4) in Fig.~\ref{FIG:NeumannTotal}, one container has all $\omega_1$-boxes and the other one
contains all the rest. Note that this procedure is possible without
performing any work as can be seen via Dalton's Law~\cite{Schwabl}:
The work needed to push the semi-permeable membrane against 
the $\omega_1$-gas can be recollected at the other
side from the moving solid wall, which is pushed by the $\omega_1$-gas
into empty space. Thus we have separated the $\omega_1$-boxes from the
rest. We repeat a similar procedure until all the $\omega_j$-gases are
separated into separate containers of volume $V$.

Next we compress the containers isothermally to the volumes $p_j V$, respectively. Denoting the pressure by $P$, and using the ideal gas law, we obtain the required work
\[
   \int_V^{p_j V} P\, {\mathrm d} V =\int_V^{p_j V} N_j k_B T/V \,{\mathrm d} V = p_j N k_B T \log p_j,
\]
where $\log$ denotes the natural logarithm.
As the temperature and thus the internal energy
remain constant, we extract heat $N k_B T \sum_j p_j \log
p_j$.

At this point, we have achieved that every container contains a pure
state $\omega_j$. We now transform every $\omega_j$ to another pure
state $\omega'$ which we choose to be the same for all containers. This is achieved by opening the boxes and applying  a reversible transformation $T_j$ in every container $j$ which satisfies $T_j\omega_j=\omega'$. These transformations exist due to Postulate 1. Since the same transformation $T_j$ is applied to all small boxes in any given container $j$ (without conditioning on the content of the small box), this operation is thermodynamically reversible.

Now we merge the containers, ending with a pure
$\omega'$-gas in the same condition as the initial
$\omega$-gas. This merging is reversible, because the density is not
changed and because all states are the same, so one can just put in
the walls again. The only step that caused an entropy difference was
the isothermal compression. Thus, the difference of the entropies between
the $\omega$-gas and the $\omega'$-gas (which are equal to the
entropies of the respective GPT
ensembles) is $N k_B \sum_j p_j \log p_j$. Therefore $H_{\omega\text{-ensemble}} =
H_{\omega'\text{-ensemble}}- N k_B \sum_j p_j \log p_j$. 
If we assume that pure states have entropy zero, we thus end up with
\begin{equation} \label{eq:Entropy}
	H_{\omega\text{-ensemble}} = - N k_B \sum_j p_j \log p_j
\end{equation}
and with the following entropy per system of the ensemble:
\begin{equation}
	H(\omega) := \frac 1 N H_{\omega\text{-ensemble}} = - k_B \sum_j p_j \log p_j.
	\label{eqVonNeumannEntropy}
\end{equation}
In summary, we have made the following assumptions to arrive at this notion of thermodynamic entropy:
\begin{assumptions}
\label{AssThoughtExp}
\begin{itemize}
	\item[(a)] Every (mixed) state can be prepared as an ensemble/statistical mixture of perfectly distinguishable pure states.
	\item[(b)] A measurement that perfectly distinguishes those pure states can be implemented as a semi-permeable membrane, which in particular does not disturb the pure states that it distinguishes.
	\item[(c)] All pure states can be reversibly transformed into each other.
	\item[(d)] Thermodynamical entropy $H$ is continuous in the state. (Since ensembles must have rational coefficients $p_j=N_j/N$, we need this to approximate arbitrary states in the thought experiment.)
	\item[(e)] All pure states have entropy zero.
\end{itemize}	
\end{assumptions}

A generalized version of the thought experiment presented by Petz~\cite{Petz} is applicable to more general decompositions: suppose that $\omega_1,\ldots, \omega_n \in \Omega_A$ are perfectly distinguishable, but not necessarily pure. Let $p_1,...,p_n$ be a probability distribution. Then Petz' thought experiment implies that
\begin{equation}
	H\left(\sum_j p_j \omega_j\right) = \sum_j p_j H(\omega_j) - k_B \sum_j p_j \log p_j.
	\label{eqPetz}
\end{equation}
The main idea is that steps 1)--5) of von Neumann's
  thought experiment can be run even if the perfectly distinguishable
  states $\omega_1,\ldots,\omega_n$ are mixed and not pure (as long as
  the membrane will still keep them undisturbed). Then the entropy of
  the state in 5) can be computed by making an additional
  \emph{extensivity assumption}: denote the GPT entropy of an
  $\omega$-ensemble of $N$ particles in a volume $V$ by
  $H_{\omega\text{-ensemble}}(N,V)$, then this assumption is that
\[
   H_{\omega\text{-ensemble}}(\lambda N,\lambda V)=\lambda\, H_{\omega\text{-ensemble}}(N,V)
\]
for $\lambda\geq 0$. Assuming in addition that the entropy of the $n$ containers adds up, the total entropy of the configuration in step 5) is $N\sum_j p_j H(\omega_j)$, from which Petz obtains~(\ref{eqPetz}). While this approach needs this additional extensivity assumption, it does not need to postulate that all pure states can be reversibly transformed into each other (in contrast to von Neumann's version). Under the assumption that all pure states have entropy zero, it reproduces eq.~(\ref{eqVonNeumannEntropy}) as a special case.

We conclude this section with a few comments on the idealizations used in the thought experiments above.
The use of gases in which the exact numbers of particles with
  each internal state is known parallels von Neumann's argument in
  \cite{Neumann}.  We rarely if ever have such precise
  knowledge of particle numbers in real physical gases, so our
  argument involves a strong idealization, but one that is common in thermodynamics and that has also been made by von
  Neumann.\footnote{Here, von Neumann's thought experiment is formulated in terms of a frequentist view on probabilities, which is standard in most treatments on thermodynamics. A treatment involving a finite ensemble
    where the frequencies (and perhaps the total particle number) are
    stochastic might seem more suitable from a Bayesian point of view; it would likely raise
    issues about whether the amount of work extracted from a finite system is subject to fluctuations. For systems that are finite or out of equilibrium, measures such as Shannon's are known not to be
  the whole story (cf. \cite{SecondLaws} and references therein).  But
  even for finite systems with a more realistic treatment of
  uncertainty about particle numbers, the von Neumann entropy still
  gives the \emph{expected} work in the protocol he considers. We defer these issues to future work, although
    we note that \cite{SecondLaws} suggests the operational entropies
    discussed in Section~\ref{SubsecEntropies} are among the relevant tools for tackling
    them.}

Although fluctuations in work  are significant for small
  particle numbers, in the thermodynamic limit of large numbers of
  particles there is concentration about the expected value 
  given, in von Neumann's protocol, by the von Neumann entropy, and
  therefore our arguments (and von Neumann's) have the most physical
  relevance in this large-$N$ situation.  This is of course true for
  classical thermodynamics as well---indeed, the use made of the ideal
  gas law and Dalton's law in von Neumann's argument are additional
  places where large $N$ is needed if one wants fluctuations to be
  negligible.  We expect finer-grained considerations to be required
  for a thorough study of fluctuations in finite systems, which is one
  reason for interest in the additional entropic measures studied in Subsection~\ref{SubsecEntropies}, but von Neumann's argument does not
  concern these finer-grained aspects of the thermodynamics of finite systems.

\section{Why the ``gbit'' is ruled out}
\label{SecGbit}
In Section~\ref{SecFramework}, we have introduced the ``gbit'', a
system for which the state space $\Omega$ is a square. Gbits are particularly interesting because
  they correspond to ``one half'' of a Popescu-Rohrlich
  box~\cite{PopescuRohrlich} which exhibits correlations that are
  stronger than those allowed by quantum theory~\cite{Popescu}. One
might wonder whether the thought experiments of
Section~\ref{SecThoughtExp} allow us to define a notion of
thermodynamic entropy for the gbit. We will now show
  that this is not the case, which can be seen as a thermodynamical
  argument for why we do not see superstrong correlations of the Popescu-Rohrlich type in our universe.

Since not all states of a gbit can be written as a mixture of perfectly distinguishable \emph{pure} states, von Neumann's original thought experiment cannot be of direct use here. However, we may resort to Petz' version: every mixed state $\omega$ of a gbit can be written as a mixture of perfectly distinguishable \emph{mixed} states, as illustrated in Figure~\ref{FIG:GbitEntropy}. Furthermore, the other crucial assumption on the state space is satisfied, too: for every pair of perfectly distinguishable mixed states, there is an instrument (a ``membrane'') that distinguishes those states without disturbing them. We even have that all pure states can be reversibly transformed into each other (namely by a rotation of the square).

Thus, we can analyze the behavior of a gbit state space in Petz'
version of the thought experiment. Any continuous notion of
thermodynamic entropy $H$ consistent with this thought experiment
would thus have to satisfy~(\ref{eqPetz}). However, we will now show
that the gbit does not admit any notion of entropy
 that
satisfies~(\ref{eqPetz}). Consider different decompositions of the state
 $\omega=\frac 1 2
\omega_a + \frac 1 2 \omega_b$ in the center of the
 square, where
$\omega_a = p \omega_1 + (1-p)\omega_2$ as well as
 $\omega_b = p
\omega_3 + (1-p) \omega_4$. It is geometrically clear
 that every
choice of $0<p<1$ corresponds to a valid decomposition. We find
(applying Eq. (\ref{eqPetz}) to $\omega$ for the first equality, 
and to $\omega_a$ and $\omega_b$ for the second): 
\begin{eqnarray*}
	H(\omega)&=& \frac 1 2 H(\omega_a) + \frac 1 2 H(\omega_b) - 2 k_B
 \frac 1 2 \log \frac 1 2 \\
	&=& \frac 1 2 p\, H(\omega_1)+\frac 1 2 (1-p) H(\omega_2)\\
	&& + \frac 1 2 p\, H(\omega_3)+\frac 1 2 (1-p) H(\omega_4)\\
	&& -k_B p\log p -k_B (1-p)\log(1-p)+k_B\log 2.
\end{eqnarray*}
This expression can never be constant in $p$, no matter what value of entropy of the four pure states $H(\omega_i)$ we assume. Thus, the entropy $H(\omega)$ of the center state $\omega$ is not well-defined, since it depends on the choice of decomposition.

\begin{figure}
\begin{center}	\includegraphics[width=0.2\textwidth]{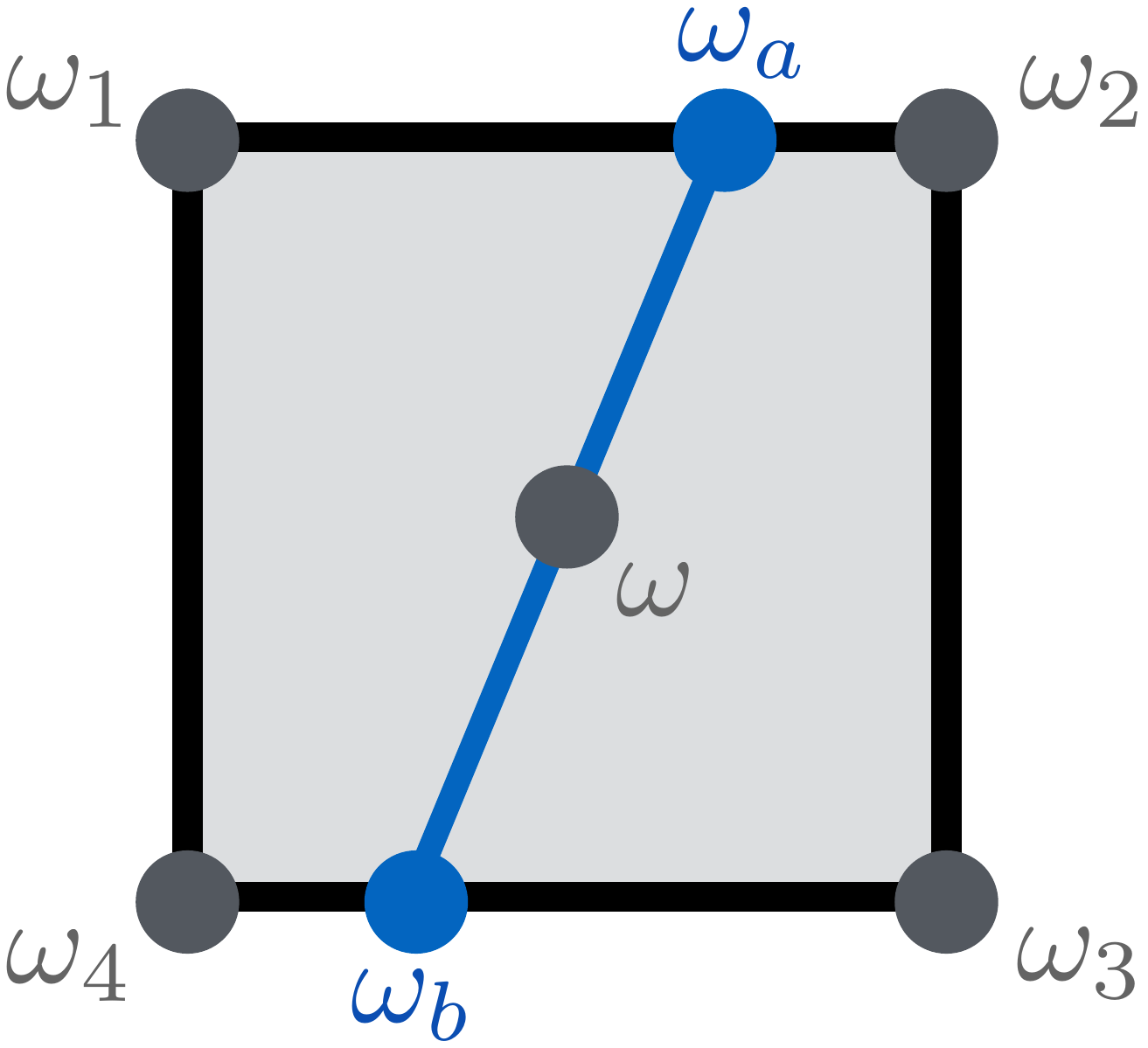}
	\caption{\small\textit{In an attempt to define a notion of thermodynamic entropy for the gbit, we can decompose any state into perfectly distinguishable states. This is done in two steps, as explained in the main text.}}
	\label{FIG:GbitEntropy}
\end{center}
\end{figure}

In other words, the structure of the gbit state space enforces that any meaningful notion of thermodynamic entropy $H$ will not only be a function of the \emph{state}, but a function of the \emph{ensemble that represents the state}. If a state $\omega$ is represented by different ensembles, then this will in general give different values of entropy.

So what goes wrong for the gbit? Clearly, all we can say with
certainty is that the combination of assumptions made in von Neumann's
thought experiment turns out not to yield a unique notion of entropy,
while a deeper physical interpretation seems only possible under
further assumptions on the interplay between the gbit and the
thermodynamic operations. However, a comparison with quantum theory
motivates at least one further speculative attempt 
 at
interpretation. In the example above, we have decomposed a state
$\omega$ into two perfectly distinguishable states $\omega_a$ and
$\omega_b$, which can themselves be decomposed into pairs of perfectly
distinguishable states $\omega_1$ and $\omega_2$, or $\omega_3$ and
$\omega_4$ respectively. In quantum theory, this would only be
possible if $\omega_a$ and $\omega_b$ are orthogonal, which would then
imply that all four states $\omega_1,\ldots,\omega_4$ are pairwise
orthogonal. This would enforce that there exists a unique projective
measurement (a ``membrane'') that distinguishes all these four states
jointly. This membrane could feature in von Neumann's thought experiment (or other similar thermodynamical settings), yielding a unique notion of thermodynamic entropy.

On the other hand, in the gbit, the four pure states
$\omega_1,\ldots,\omega_4$ are \emph{not} jointly perfectly
distinguishable. Hence there is no canonical choice of ``membrane''
that could be used in the thought experiment to define a unique
natural notion of entropy for the gbit states. Entropy will be
``contextual'', depending on the choice of membrane resp.\ ensemble decomposition that is used in any given specific thermodynamical setting. Therefore, the implication \emph{``pairwise distinguishability$\Rightarrow$joint
  distinguishability''}, which is true for quantum theory, has thermodynamic relevance. This implication, if suitably interpreted, leads to the ``exclusivity principle''~\cite{Acin, LSW11,
    Cabello2012}, namely that the sum of the probabilities of pairwise exclusive propositions cannot exceed $1$ (in this case these propositions correspond to the outcomes of the jointly distinguishing measurement). This suggests that the exclusivity principle, which has so far been considered only in the realm of contextuality, may be thermodynamically relevant. This observation is also closely related to the notion of ``dimension mismatch'' described in~\cite{Brunner}, and to orthomodularity in quantum logic (see for example~\cite{CassinelliLahti}).

\section{A class of theories with consistent thermodynamic behavior}

\subsection{The two postulates}
\label{SubsecPostulates} 
In this section we introduce the two postulates that express key operational
concepts from thermodynamics. The first postulate is motivated by the
universality of thermodynamics and the distinction between microscopic
and macroscopic behaviour. At first we consider the universality of
thermodynamics, in the sense that thermodynamics is a very general
theory whose basic principles can be applied to many possible
implementations, as already noticed by N.\ Carnot \cite{Carnot}:
 
\textit{``In order to consider in the most general way the principle
  of the production of motion by heat, it must be considered
  independently of any mechanism or any particular agent. It is
  necessary to establish principles applicable not only to steam
  engines but to all imaginable heat-engines, whatever the working
  substance and whatever the method by which it is operated.''}

Recalling von Neumann's thought experiment in the case of quantum theory, we can think of thermodynamical protocols (which will ultimately also include heat engines) as acting on a given ensemble, defined as a probabilistic mixture of pure states chosen from a fixed basis. If we interpret ensembles with different choices of basis as different ``working substances'', then Carnot's principle should apply: protocols that can be implemented on one ensemble (say, ensemble 1) can also be implemented on the other (say, ensemble 2).\footnote{Here we only consider ensembles of identical Hilbert space dimensions. If the dimensions are different (say, $2$ versus $3$), then one can implement different sets of protocols on the ensembles (say, ones involving semipermeable membranes that distinguish $3$ alternatives in the latter, but not the former case). One could then still discuss a notion of universality in Carnot's spirit, by referring to the equivalence of, say, a state space with $N=3$ alternatives to a subspace of a state space with $N=2\times 2$ alternatives, but we will not discuss this further here.} In quantum theory, this universality is ensured by the existence of unitary transformations: all orthonormal
       bases can be translated into each other by a unitary and
       therefore reversible map. In this sense, the state of ensemble 1 can in principle be transferred to ensemble 2, then the thermodynamic protocol of ensemble 2 can be performed (if we have also transformed the projectors describing the membranes accordingly), and then one can transform back. Even if this cannot always be achieved in practice, the corresponding unitary symmetry of the quantum state space (considered as passive transformations between different descriptions) enforces the aforementioned universality.\footnote{In classical thermodynamics, the analog of a choice of basis is the labelling of the distinguishable configurations. Clearly, the availability of thermodynamic protocols does not change under relabelling.}
 
 This universality of implementation, as well as independence of the choice of labels and descriptions, should continue to hold in all generalized
       theories that we consider. An orthonormal basis from quantum
       theory is nothing else than a set of perfectly distinguishable
       pure states, i.e.\ an $n$-frame. Therefore, in our generalized theories, we expect that this universality of implementation is achieved by the existence of reversible transformations that, in analogy to unitary maps, transform any given $n$-frame into any other:\\
 
       \textbf{Postulate 1:} For each $n \in \mathbb N$, all sets of
       $n$ perfectly distinguishable pure states are equivalent. That
       is, if $\{\omega_1, \ldots, \omega_n\}$ and
       $\{\varphi_1,...,\varphi_n \}$ are two such sets, then there
       exists a reversible transformation $T$ with $T \omega_j =
       \varphi_j$ for all $j$.
 
       Furthermore, Postulate 1 expresses a physical property that is
       crucial for thermodynamics: that of \emph{microscopic
         reversibility}. Many characteristic properties of
       thermodynamics arise from limited experimental access to the
       microscopic degrees of freedom, which by themselves undergo
       reversible time evolution. This reversibility, for example,
       forbids evolving two microstates into
       one, which is at the heart of the non-decrease of entropy. If
       the experimenter had full access to the microscopic degrees of
       freedom, then he or she could convert any state of maximal
       knowledge to any other one as long as he or she preserved distinguishability. Postulate 1 formalizes this
       microscopic basis of thermodynamics by demanding the existence
       of ``enough'' distinguishability-preserving, microscopic
       transformations $T$, which can be understood as reversible time
       evolutions.

Postulate 1 has substantial information-theoretical justifications and consequences. The basic concepts of both thermodynamics and information processing are independent of the choice of implementation. For information processing this is formalized by the Turing machine which admits a multitude of physical realizations. Perfectly distinguishable pure states can be taken as bits, and Postulate 1 expresses that all bits (or their higher-dimensional analogues) are equivalent. It is for this reason that Postulate 1 was called \textit{generalized bit symmetry} in~\cite{MariusMaster}, and its restriction to pairs of distinguishable states was called \textit{bit symmetry} in~\cite{Skalarprodukt}. Starting with Landauer's principle, ``thermodynamics of computation''~\cite{Bennett} has become a fruitful paradigm that relates the two apparently disjoint fields. The two complementary interpretations of Postulate 1 are one instance of this.

Now we turn to our second postulate. We are looking for theories very similar to the thermodynamics we are used to;
  thus it is essential that we can adopt basic notions of standard
  thermodynamics  unchanged or with only very small alterations. Two such
  notions of great importance are (Shannon) entropy $S = -k_B
  \sum_j p_j \log p_j$ and majorization theory. In classical and
  quantum thermodynamics, these notions operate on the coefficients in 
  a decomposition of a state into perfectly distinguishable pure states (in quantum theory, the eigenvalues). In 
  order to not change thermodynamic theory too much, we would also like this
  to be possible in our more general state spaces. Thus, we demand that every state
  has a convex decomposition into perfectly distinguishable pure
  states.
  
Note that this was indeed one of our assumptions in von
  Neumann's thought experiment in Section~\ref{SecThoughtExp}. There,
  it allowed us to realize any state $\omega$ as a ``quasiclassical
  ensemble'', i.e.\ as an ensemble of states that behave like
  classical labels. This gives us a further justification of our
  second postulate: thermodynamic (thought) experiments require that
  states have an ensemble interpretation.  An unambiguous notion of
  ``counting of microstates'' demands that the ensembles consist of
  perfectly distinguishable, pure states.  Without this,  
obtaining a phenomenological thermodynamics for which the theory is 
the underlying microscopic theory seems problematic.
Thus, our second postulate is\\

\textbf{Postulate 2}: Every state $\omega \in \Omega_A$ has a convex decomposition $\omega= \sum_j p_j \omega_j$ into perfectly distinguishable pure states $\omega_j$.\\

It is tempting to interpret the two postulates as reflecting the \emph{microscopic} and the \emph{macroscopic} aspects of thermodynamics, respectively: while Postulate 1 describes microscopic reversibility of the pure states that may describe single particles in thermodynamics, Postulate 2 ensures that mixed states can be interpreted macroscopically as descriptions of quasiclassical ensembles, composed of a large number of particles that are separately in unknown but distinguishable microstates.

We will not introduce any further postulates. In particular, we will
not make any assumptions on the \emph{composition} of systems. All our
results are therefore independent from notions like \emph{tomographic
  locality}~\cite{Hardy} (which is arguably dispensable in many important
situations~\cite{BarnumCategories}) or \emph{purification}~\cite{ChiribellaPurification} (which is a
rather strong assumption); we do not assume either of the two.

\subsection{Some consequences of Postulates 1 and 2}
\label{SubsecConsequences}
Postulates 1 and 2 have been analyzed in~\cite{Postulates}, but in a 
different context: instead of  investigating thermodynamics, the goal in~\cite{Postulates} was to
  obtain a reconstruction of quantum theory, by supplementing
  Postulates 1 and 2 with further postulates. Some of the insights
  from~\cite{Postulates} will be important here, and are 
therefore briefly
  discussed below. Starting with Subsection~\ref{SubsecObservables},
  we will also obtain new results that are interesting in a
  thermodynamic context.

 In contrast to Hilbert space,
there is no apriori notion of inner product for GPTs. However, as
shown in~\cite{Skalarprodukt}, we get a natural inner product
$\langle\cdot,\cdot\rangle$ as a consequence of Postulates 1 and 2: it
satisfies $\braket{T\omega,T\varphi} = \braket{\omega,\varphi}$ for
all reversible transformations $T$, and $0 \le \braket{\omega,\varphi}
\le 1$ for all states $\omega,\varphi \in \Omega_A$. Furthermore,
$\braket{\omega,\omega} = 1$ for all pure $\omega\in \Omega_A$ and
$\braket{\varphi,\varphi} < 1$ for all mixed $\varphi \in \Omega_A$,
and $\braket{\omega,\varphi} = 0$ if $\omega,\varphi \in \Omega_A$ are
perfectly distinguishable. Thus, all perfectly distinguishable states
are orthogonal, as in quantum theory.

 Moreover, the cone of unnormalized states becomes \emph{self-dual}
with this choice of inner product. In particular, every effect $e$ can
be taken as a vector in $A_+$, such that $e(\omega)=\langle
e,\omega\rangle$.
 In standard quantum theory, this is the Hilbert
Schmidt inner product on the real vector space of Hermitian matrices:
$\langle X,Y\rangle=\tr(XY)$ for $X=X^\dagger$, $Y=Y^\dagger$.

Quantum theory has more structure: the convex set of density matrices
$\Omega_A$ has faces\footnote{A \emph{face} of a convex set $C$ is a convex subset $F\subseteq C$ with the property that $\lambda x+(1-\lambda)y\in F$ with $0<\lambda<1$ and $x,y\in C$ implies $x,y\in F$~\cite{Webster}. We say that $F$ is \emph{generated by} $\omega_1,\ldots,\omega_n$ if $F$ is the smallest face that contains $\omega_1,\ldots,\omega_n$.}, and these faces are in one-to-one correspondence
to subspaces of Hilbert space (namely, a face $F$ contains all density
matrices that have support on the corresponding Hilbert subspace). To
every face $F$, we can associate a number $|F|$ which is the dimension
of the corresponding Hilbert subspace, and $F\subsetneq G$ implies
$|F|<|G|$. Every face $F$ is generated by $|F|$ pure and perfectly
distinguishable states in $F$ (an $|F|$-frame in $F$), and every (smaller) frame that is a subset of $F$ can be completed, or extended, to a frame which has
$|F|$ elements and thus generates $F$.
 
In all theories that
satisfy Postulates 1 and 2, all these properties hold in complete
analogy~\cite{Postulates}. However, since faces do not any more correspond to Hilbert
spaces, the numbers $|F|$ do not have an interpretation as the dimension of a subspace. Instead, we call $|F|$ the \emph{rank} of $F$.  
If von Neumann's thought experiment is supposed to make sense for
these theories, we need a way to formalize the working of a
semipermeable membrane, which in quantum theory is done via projective measurements.
 
Since we are dealing with unnormalized states, the corresponding analog in GPTs
will be formulated in terms of the set of unnormalized states
$A_+$. As one can see in the case of the gbit, it is not automatic
that we have any notion of ``projective measurements'' for any given state space. However, Postulates 1 and 2 turn out to ensure
that projective measurements exist. For any face $F$ of $A_+$ (the
non-negative multiples of the corresponding face of $\Omega_A$),
consider the orthogonal projector $P_F$ onto the span of $F$. One can
show that $P_F$ is \emph{positive}, i.e.\ maps (unnormalized) states
to (unnormalized) states~\cite{Postulates}. Moreover, $P_F$ does not disturb the states
in the face $F$.
 
 Thus, to a given set of mutually orthogonal
faces $F_1,\ldots, F_m$ such that $|F_1|+\ldots+|F_m|=N_A$, we can
associate an instrument with transformations $T_i:=P_{F_i}$, which
describes a projective measurement, as in a semipermeable
membrane. Transformation $T_i$ leaves the states in face $F_i$
unperturbed, but fully blocks out states in the other faces, i.e. $T_i
\omega=0$ for $\omega\in F_j$, $i\neq j$. In standard quantum theory,
these transformations are $P_{F_i}\rho= \pi_i \rho \pi_i$, where
$\pi_i$ is the orthogonal Hilbert space projector onto the $i$-th
Hilbert subspace. The rank condition becomes
$\tr(\pi_1)+\ldots+\tr(\pi_m)= N_A$ (the total Hilbert space
dimension), and mutual orthogonality is
$\pi_i\pi_j=\delta_{ij}\pi_i$. We will show in Subsection
\ref{SubsecObservables} that the mutually orthogonal faces
replace the eigenspaces from quantum theory and that the projective
measurement described here can be interpreted as measuring an
observable.\\

 The Hilbert space projector $\pi_i$ therefore
also has an interpretation as an \emph{effect} in standard quantum
theory: it yields the probability of outcome $i$ in the projective
measurement on a state $\rho$, namely $\tr(\pi_i\rho)$. The analogous effect in a GPT that satisfies Postulates 1
and 2, corresponding to a face $F$, is
\[
   u_F:= P_F u_A
\]
(identifying the effect $u_A$ with a vector via the inner product). The effect $u_F$ is sometimes called the ``projective unit'' of $F$. In quantum theory, we can write $\pi_i = \sum_j |\psi_j\rangle\langle\psi_j|$, where the $|\psi_j\rangle$ are an orthonormal basis of the corresponding Hilbert subspace. The same turns out to be true in our GPTs: we have
\begin{equation}
   u_F=\sum_{j=1}^{|F|} \omega_i,
   \label{eqDecompUnit}
\end{equation}
where $\omega_1,\ldots,\omega_{|F|}$ is any frame that generates $F$. Therefore, the probability to obtain outcome $i$ in the projective measurement above on state $\omega$ is $\langle u_{F_i},\omega\rangle=\langle u_A,P_{F_i} \omega\rangle$.

\subsection{State spaces satisfying Postulates 1 and 2}
It is easy to see that both quantum and classical state spaces satisfy Postulates 1 and 2. By a ``classical state space'', we mean a state space that consists of discrete probability distributions. Concretely, for any number $N\in\mathbb{N}$ of mutually exclusive alternatives, consider the state space
\[
   \Omega:=\{(p_1,\ldots,p_N)\,\,|\,\,p_i\geq 0,\enspace \sum_i p_i=1\}.
\]
Any pure state is given by a deterministic probability vector, i.e.\ $\omega_i=(0,\ldots,0,1,0,\ldots,0)$ (where $1$ is on the $i$-th place). If we have two equally sized sets of such vectors (as in Postulate 1), then there is always a permutation that maps one set to the other. In fact, the reversible transformations correspond to the permutations of the entries. Postulate 2 is then simply the statement that
\[
   (p_1,\ldots,p_N)=p_1 \omega_1+p_2\omega_2+\ldots+p_N\omega_N.
\]
Which state spaces are there, in addition to standard complex quantum theory and classical probability theory, that satisfy Postulates 1 and 2? We think
that this question is very difficult to answer. Thus, we formulate the following\\

\noindent
\textbf{Open Problem 1.} \textit{
	Classify all state spaces that satisfy Postulates 1 and 2.}\\

From the results in~\cite{Postulates}, we know which state spaces
satisfy Postulates 1 and 2 \emph{and} one additional property: the
absence of third-order interference. The notion of higher-order
interference has been introduced by Sorkin~\cite{Sorkin}, and has
since been the subject of intense
theoretical~\cite{UdudecBarnumEmerson,Henson,LeeSelby1} and
experimental~\cite{Sinha,Sinha2009,Soellner,Kauten,Hickmann,Park}
interest.

For the main idea, think of three mutually exclusive alternatives in quantum theory (such as three slits in a triple-slit experiment), described by orthogonal projectors $\pi_1,\pi_2,\pi_3$. The event that alternative 1 \emph{or} alternative 2 takes place is described by the projector $\pi_{12}=\pi_1+\pi_2$; similarly, we have $\pi_{13}$, $\pi_{23}$ and $\pi_{123}$. Their actions on density matrices are described by superoperators
\[
   \rho\mapsto P_{12}(\rho):=\pi_{12} \rho \pi_{12}
\]
(and similarly for the other projectors). As a consequence, we obtain that $P_{12}\neq P_1+P_2$, which expresses the phenomenon of \emph{interference}. However, it is easy to check that
\begin{equation}
	P_{123}=P_{12}+P_{13}+P_{23}-P_1-P_2-P_3,
	\label{eqNoThirdOrder}
\end{equation}
which means that interference over three alternatives can be reduced to contributions from interferences of pairs of alternatives. Similar identities hold for an arbitrary number $n\geq 4$ of alternatives: \emph{quantum theory admits only pairwise interference}, and no ``third-order interference'' which would be characterized by a violation of this equality.

In the context of Postulates 1 and 2, we have an analogous notion of orthogonal projectors, and thus we can consider~(\ref{eqNoThirdOrder}) and its generalization to $n\geq 4$ alternatives on a state space with $N\geq n$ perfectly distinguishable states. Postulating this ``absence of third-order interference'' in addition to Postulates 1 and 2 gives us the following:

\begin{thm}[Lemma 33 in~\cite{Postulates}]
\label{LemList}
	The possible state spaces which satisfy Postulates 1 and 2 \emph{and} which do not admit third-order interference, in addition to classical state spaces, are the following. First, for $N\geq 4$ perfectly distinguishable states, there are only three possibilities:
	\begin{itemize}
		\item Standard complex quantum theory.
		\item Quantum theory over the real numbers. That is, only real entries are allowed in the $N\times N$ density matrices.
		\item Quantum theory over the quaternions. The state spaces are the self-adjoint $N\times N$ quaternionic matrices of unit trace.
	\end{itemize}
	For $N=3$ perfectly distinguishable states, all of the above \emph{and} one exceptional solution are possible, namely quantum theory over the octonions (but only for the case of $3\times 3$ unit trace density matrices).
	
	For $N=2$ (the ``bit'' case), we have the $d$-dimensional Bloch ball state spaces $\Omega_d := \{(1,r)^T|r\in \mathbb R^d, \|r\|\le 1\}$ with $d\geq 2$. They are analogous to the standard Bloch ball $\Omega_3$ of quantum theory, with very similar descriptions of effects etc. Their group of reversible transformations may either be ${\rm SO}(d)$ (which corresponds to ${\rm PU}(2)$ for $d=3$), or some subgroup of ${\rm O}(d)$ which is transitive on the sphere (such as ${\rm SU}(2)$ for $d=4$).
\end{thm}
Mathematically, these examples correspond to the state spaces of the finite-dimensional irreducible formally real Jordan algebras~\cite{ASBook,Postulates}.
We do not know whether there are theories that satisfy Postulates 1
and 2 but admit higher-order interference and therefore do not
appear on this list. In Theorem~\ref{ThmAllEquiv}, we will show that the question whether a theory 
has third-order interference is related to the properties of its R\'enyi entropies.

\subsection{Observables and diagonalization}
\label{SubsecObservables}
A central part of physics are observables and how they can be
measured. In standard quantum theory, we can introduce observables in
two different ways, which both equivalently lead to the prescription
that \emph{observables are described by Hermitian operators/matrices}. 

First, in finite dimensions, we can characterize observables as \emph{those objects that
 linearly assign real expectation values to states}. In the case of quantum theory it follows
that observables are represented by matrices $X$, and Hermiticity
$X=X^\dagger$ implies that expectation values $\tr(\rho X)$ are always
real. Linearity is enforced by the statistical interpretation of states, for the same reason that effects in GPTs are linear.

Second, we can introduce observables by saying that there is a projective measurement $\pi_1,\ldots,\pi_n$ that measures this observable, and which has outcomes $x_1,\ldots,x_n\in\R$. This leads to the Hermitian operator $X=\sum_{i=1}^n x_i \pi_i$. Since every Hermitian operator can be diagonalized, these two definitions are equivalent.

Our two postulates provide the structure to introduce observables in a completely analogous way. First, using the inner product, we can define observables as linear maps of the form
\[
   \omega\mapsto \langle x,\omega\rangle
\]
and thus identify them with elements $x\in A$ of the vector space that carries the states (as in quantum theory, where this vector space is the space of Hermitian matrices). As noticed in~\cite{UdudecDoktor}, every such vector has a representation of the form
\begin{equation}
   x=\sum_i x_i u_i,
   \label{eqDiag}
\end{equation}
where the $u_i$ are projective units corresponding to mutually orthogonal faces $F_i$, $x_i\in\R$, and $x_i\neq x_j$ for $i\neq j$. The analogy with quantum theory goes even further: due to~(\ref{eqDecompUnit}), we have
$x=\sum_i x_i \sum_j \omega_i^{(j)}$, whenever $\omega_i^{(1)},\ldots,\omega_i^{(|F_i|)}$ is a frame on $F_i$. This corresponds to the identity $X=\sum_i x_i \sum_j |\psi_i^{(j)}\rangle\langle\psi_i^{(j)}|$ in standard quantum theory. In analogy to quantum theory we will call the $F_i$ eigenfaces and the $x_i$ eigenvalues.
To further justify this terminology, note that the $x_i$ are eigenvalues of the map $\sum_i x_i P_i$, where $P_i$ are the orthogonal projectors onto the spans of the faces $F_i$.

\begin{thm}
\label{theorem:Observables}
If Postulates 1 and 2 hold, then every element $x \in A$ has a representation of the form $x=\sum_{j=1}^n x_j u_j$ where $x_j \in \mathbb R$ are pairwise different and the $u_j$ are the projective units of pairwise orthogonal faces $F_j$ such that $\sum_j u_j = u_A$. This decomposition $x=\sum_{j=1}^n x_j u_j$ is unique up to relabelling. In analogy to quantum theory, we will call the $x_j$ eigenvalues and the $F_j$ eigenfaces.

Furthermore, for every real function $f$ with suitable domain of definition, we can define
\begin{equation}
   f(x):=\sum_{j=1}^n f(x_j) u_j
   \label{eqSpectralCalculus}
\end{equation}
as in spectral calculus.

If $P_j$ is the orthogonal projector onto the span of $F_j$, then $(P_1,...,P_n)$ is a well-defined instrument with induced measurement $(u_1,...,u_n)$ which leaves the elements of $\mathrm{span}(F_j)$ invariant:
\[
	P_j(\omega) = \delta_{jk} \cdot \omega \quad \mbox{for all } \omega \in F_k.
\]
In analogy to quantum theory, we will call this  instrument the projective measurement of the observable $x$.
\end{thm}	

We will give a proof in the appendix.\footnote{This can also also obtained by
  combining the fact that Postulates 1 and 2 imply the state space is
  projective (first part of Theorem 17 in \cite{Postulates}) and
  self-dual (Proposition 3 in \cite{Postulates}) with results such as
  Theorem 8.64 in \cite{ASBook}.}
Eq.~(\ref{eqSpectralCalculus}) allows us to define a notion of entropy, in full analogy to quantum mechanics.

\begin{definition}[Spectral entropy]
\label{DefSpectralEntropy}
If $A$ is a state space that satisfies Postulates 1 and 2, we define the \emph{spectral entropy} for any state $\omega\in\Omega_A$ as
\[
   S(\omega):=-\sum_i p_i \log p_i,
\]
where $\omega=\sum_i p_i \omega_i$ is any convex decomposition of $\omega$ into pure and perfectly distinguishable states $\omega_i$, and $0\log 0:=0$.
\end{definition}
Theorem~\ref{theorem:Observables} tells us that this definition is independent of the choice of decomposition: it is easy to check that
\[
   S(\omega)=-\langle \omega,\log\omega\rangle,
\]
where $\log\omega$ is understood in the sense of spectral calculus as in~(\ref{eqSpectralCalculus}). The right-hand side is manifestly independent of the decomposition. It can also be written $S(\omega)=u_A(\eta(\omega))$, where $\eta(x)=-x\log x$ for $x>0$ and $\eta(0)=0$. In particular,
\begin{equation}
\omega\mbox{ is a pure state }\Leftrightarrow S(\omega)=0.	
\label{eqPureZero}
\end{equation}
To see this, note that any pure state $\omega_1=\omega$ can be
extended to a set of perfectly distinguishable pure states
$\omega_1, \omega_2,\ldots,\omega_{N_A}$ such that $\omega=1\cdot
\omega_1+0\cdot\omega_2+\ldots+0\cdot\omega_{N_A}$. Conversely, if
$S(\omega)=0$, then any decomposition of $\omega$ must have
coefficients $(1,0,\ldots,0)$.

\subsection{Thermodynamics in the context of Postulates 1 and 2}
If a state space satisfies Postulates 1 and 2, then it
  also satisfies all the assumptions that we have made in von
  Neumann's thought experiment. It is easy to check all items in
  Assumptions~\ref{AssThoughtExp}: (a) is simply Postulate 2, and
  (c) is a consequence of Postulate 1. As we have seen in the
  previous section, our two postulates imply that we have orthogonal
  projectors sharing important 
properties with those of standard quantum theory. If we
  make the physical assumption that we can actually implement them by
  means of semipermeable membranes (as in quantum theory), we obtain
  (b). Item (e) is the same as~(\ref{eqPureZero}). Note that
  assumption (d) is not a mathematical assumption about the state space, but a physical assumption
  about thermodynamic entropy. This shows part of the following (the full proof will be given in the appendix):
\begin{observation}
\label{ObsEntropy}
	Von Neumann's thought experiment, as explained in Section~\ref{SecThoughtExp}, can be run for every state space that satisfies Postulates 1 and 2. The notion of thermodynamic entropy $H$ that one obtains from that thought experiment turns out to equal spectral entropy $S$ as given in Definition~\ref{DefSpectralEntropy},
	\[
	   H(\omega)=S(\omega)\qquad\mbox{for all states }\omega.
	\]
	This is consistent with Assumptions~\ref{AssThoughtExp}. Furthermore, it is also consistent with Petz' version of the thought experiment, because spectral entropy satisfies
	\begin{equation}
	   S(\omega)=\sum_j p_j S(\omega_j)-\sum_j p_j \log p_j
	   \label{eqSDecomp}
	\end{equation}
	for every convex decomposition $\omega=\sum_j p_j\omega_j$ of $\omega$ into perfectly distinguishable, not necessarily pure states $\omega_j$.
\end{observation}
Thus, spectral entropy $S$ gives meaningful and consistent physical predictions in situations like von Neumann's and Petz' thought experiments. However, we clearly do not know whether $S$ is a consistent notion of physical entropy in \emph{all} thermodynamical situations.

It turns out that there are further properties of $S$ that encourage
its physical interpretation as a thermodynamical entropy. In
particular, we will now show that the \emph{second law} holds in two
important situations.  We start by considering projective measurements
$P_1,\ldots,P_n$. Projective measurements can model semipermeable
membranes as in von Neumann's thought experiment, or they describe the
measurement of an observable as explained in
Subsection~\ref{SubsecObservables}. Consider the action of this
measurement on a given state $\omega$. With probabilities $\left(u_A
\circ P_j\right)(\omega)$, this measurement yields the outcome $j$
with post-measurement state $\omega_j:=P_j \omega/\left(u_A \circ
P_j(\omega)\right)$. Performing this measurement on every particle of
an ensemble (without learning the outcomes) yields a new ensemble, described by the
post-measurement state
\[
	\omega' = \sum_{ j:\,\, u_A \circ P_j(\omega) \ne 0  } \big(u_A \circ P_j\big)(\omega) \cdot \omega_j
	= \sum_j P_j \omega.
\]

Projective measurements 
 do not decrease the entropy of the ensemble:

\begin{thm}
\label{theorem:SecondLaw}
	Suppose Postulates 1 and 2 are satisfied. Let $P_1,\ldots,P_n$ be orthogonal projectors which form a valid instrument. Then the induced measurement with post-measurement ensemble state $\omega'=\sum_j P_j w$ does not decrease entropy: $S(\omega')\ge S(\omega)$.
\end{thm}
The proof will be given in the appendix. As in standard quantum
  theory, projectors $P_j$ form a valid instrument if and only if
they are mutually orthogonal, i.e.\ $P_i P_j = \delta_{ij}P_i$, and
complete: $\sum_i u_A\circ P_i = u_A$.

\begin{figure*}
\begin{center}\includegraphics[width=0.64\textwidth]{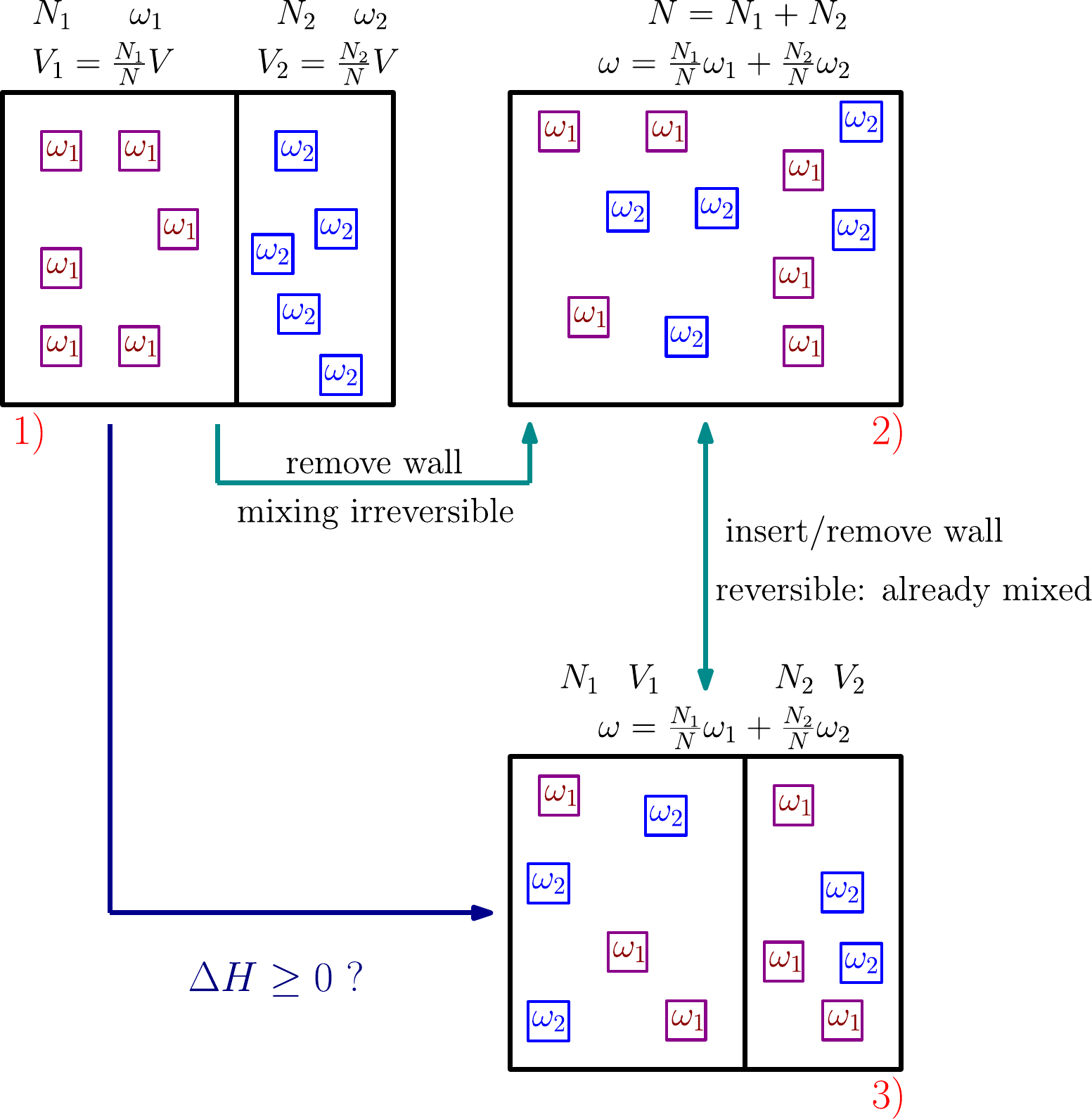}
	\caption{\small \textit{A process mixing gases by removing a separating wall. Theorem~\ref{thm:EntropyConcave} ensures that this process does not decrease entropy, i.e.\ $\Delta H\geq 0$, if thermodynamic entropy $H$ is identified with spectral entropy $S$ as suggested by von Neumann's thought experiment.}}
	\label{FIG:Mixing}
\end{center}
\end{figure*}

Another important manifestation of the second law is in mixing procedures as in Figure~\ref{FIG:Mixing}.
Consider tanks that are separated by walls. Similarly to von Neumann's thought experiment, let the $j$-th tank contain an $N_j$ -particle gas that represents an $\omega_j$-ensemble. Furthermore, assume that all the gases are at the same pressure and density. Identifying thermodynamic entropy $H$ with spectral entropy $S$ (as suggested by Observation~\ref{ObsEntropy}), the entropy of the GPT-ensemble in tank $j$ is $N_j S(\omega_j)$, where $S$ is the entropy per system. Thus the total GPT-entropy is $\sum_j N_j S(\omega_j)$. We remove the walls and let the gases mix. Then we put the walls back in. Now all the tanks contain gases hosting $\sum_j \frac{N_j}N \omega_j$ ensembles at the same conditions as before, where $N = \sum_j N_j$. The total GPT-entropy in the end is given by $\sum_ j N_j S\left(\sum_ k \frac{N_k}N \omega_k\right) = N S\left(\sum_ k \frac{N_k}N \omega_k\right)$. As the gases in the tanks have the same density, volume, temperature and pressure as before, the only difference in entropy is due to the GPT-ensembles. The second law requires that the entropy does not decrease in this process, i.e.\ that $\sum_j N_j S(\omega_j) \le N S\left(\sum_ j \frac{N_j}N \omega_j\right)$ and thus $\sum_j \frac{N_j}N S(\omega_j) \le S\left(\sum_ j \frac{N_j}N \omega_j\right)$. The following theorem shows that our two postulates guarantee that this is true:
\begin{thm} \label{thm:EntropyConcave}
	Assume Postulates 1 and 2. Then entropy is concave, i.e.\ for $\omega_1,\ldots,\omega_n \in \Omega_A$ and $p_1,...,p_n$ a probability distribution, we have
\begin{equation}
		S\left(\sum_j p_j w_j\right) \ge \sum_j p_j S(w_j).
		\label{eqConcavity}
	\end{equation}
\end{thm}

Thus, the second law automatically holds for mixing processes. One way to prove~(\ref{eqConcavity}) is to see that $S$ equals ``measurement entropy'' as we will show in Subsection~ \ref{SubsecEntropies}, proven to be concave in~\cite{WehnerEntropy} and~\cite{BarnumEntropy}. However, there is a simpler proof that uses a notion of \emph{relative entropy}, which is an important notion in its own right.
\begin{defn}
   For state spaces $A$ that satisfy Postulates 1 and 2, we define the \emph{relative entropy} of two states $\omega,\varphi\in\Omega_A$ as
\[
			S(\omega\|\varphi) := -S(\omega)-\braket{\omega, \log \varphi}.
\]
	Here, for $\varphi=\sum_j q_j \varphi_j$ any decomposition into a maximal frame, $\log \varphi = \sum_j \log(q_j) \varphi_j$ according to Theorem~\ref{theorem:Observables}. (As in quantum theory, this can be infinite if there are $q_j=0$ such that $\langle \omega,\varphi_j\rangle\neq 0$).
\end{defn}

A notion of relative entropy in GPTs has also been defined in Scandolo's Master thesis~\cite{ScandoloMaster}, but under different assumptions, as discussed in the introduction.
Relative entropy continues to satisfy \emph{Klein's inequality}, a fact that is useful in proving Theorem~\ref{thm:EntropyConcave}. The proof is similar to that within standard quantum theory and deferred to the appendix.

\begin{thm}[Klein's inequality]
\label{theorem:Klein}
	For all $\omega,\varphi \in \Omega_A$,
\[
	S(\omega \|\varphi ) \ge 0 .
\]
\end{thm}
Klein's inequality can be used to give a simple proof of Theorem~\ref{thm:EntropyConcave}:
\begin{eqnarray*}
	0 &\leq & \sum_j p_j S\big(\omega_j \| \sum_k p_k \omega_k\big)
    \\
	& =& -\sum_j p_j S(\omega_j ) - \left\langle\sum_j p_j \omega_j, \log \left( \sum_k p_k \omega_k \right)\right\rangle \\ &=& -\sum_j p_j S(\omega_j ) + S\left(\sum_k p_k \omega_k\right).
\end{eqnarray*}
Given all the calculations in this subsection in terms of
  orthogonal projections, it may seem at first sight as if every
  statement or calculation in quantum theory can be analogously made
  in the more general state spaces that satisfy Postulates 1 and
  2. However, this may not quite be true, as the fact that the  following  is an open problem shows:\\
  
\noindent
\textbf{Open Problem 2.} \textit{For state spaces satisfying Postulates 1 and 2, if $\omega$ is a pure state, and $P$ an orthogonal projection, then is $P\omega$ also (up to normalization) a pure state?}\\

In classical and quantum state spaces, the answer is ``yes'', but we do not know if a positive answer follows from Postulates 1 and 2 alone. We will return to this problem in Theorem~\ref{ThmAllEquiv}.

Note that Chiribella and Scandolo have applied similar techniques and found beautiful results, including some which are comparable to some of ours, in~\cite[Section 7]{diagonal} (see also~\cite{ScandoloMaster}). They derive diagonalizability of states from a very different set of postulates.

\subsection{Information-theoretic entropies and their relation to physics}
\label{SubsecEntropies}
So far we have considered entropy from a thermodynamic perspective. But entropies also arise in information theory, and as the GPT framework is mostly studied in quantum information theory, indeed there have been many results on entropy from a information-theoretic perspective. Our exposition will mainly follow~\cite{WehnerEntropy}, but has also been given in a slightly different formalism in~\cite{BarnumEntropy}.

Let $e=(e_1,...,e_n)$ and $f=(f_1,...,f_m)$ be two measurements such that there exists a map $M: \{1,...,n\} \rightarrow \{1,...,m\}$ with 
\[
	\sum_{\{ j| M(j) = k \}} e_j = f_k  \qquad (k=1,\ldots,m).
\]
If $M$ is bijective, then the measurement $f$ is simply a \emph{re-labelling} of $e$. If there exists a $k$ with $M(j) \ne k \ \ \forall j$, then because of the normalization of the $e$-measurement, $f_k = 0$, i.e.\ $f_k$ corresponds to a trivial outcome that never happens. If $M$ is not injective, then $f$ is a \emph{coarse-graining} of $e$ (or vice versa, $e$ a \emph{refinement} of $f$) in the sense that $f$ is obtained from $e$ by collecting several outcomes of $e$ and giving them a common new outcome label (and by possibly adding the $0$-effect a few times), see Figure~\ref{FIG:FineGraining}. In this sense, we do not care about which of the $e_j$ triggered the new effect.

\begin{figure}
\begin{center}
\includegraphics[width=0.5\textwidth]{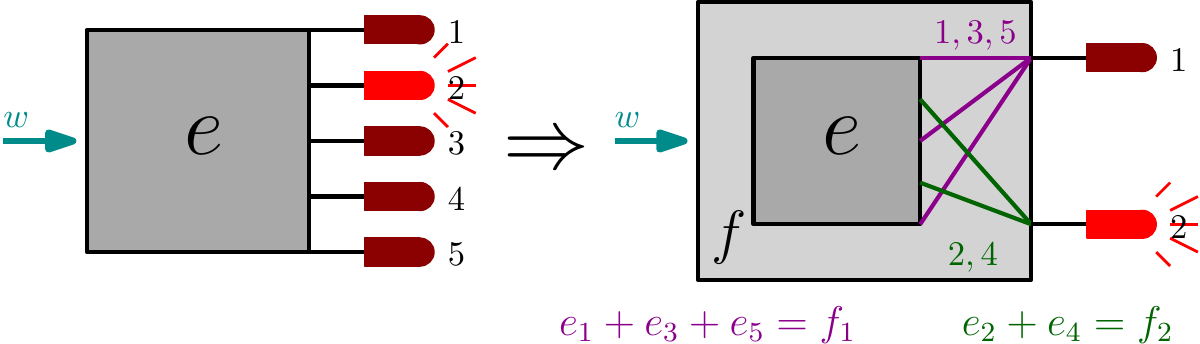}
	\caption{\small \textit{A coarse-graining of a measurement is created by having several measurement results trigger the same output.}}
	\label{FIG:FineGraining}
\end{center}
\end{figure}

However, there exist \emph{trivial} refinements/coarse-grainings: for those, $e_j \propto f_{M(j)} \ \ \forall j$. We write $e_j = p_{j} f_{M(j)}$. Then such a measurement can be obtained by performing $f$, and if outcome $k$ is triggered, we activate a classical random number generator which generates the final outcome $j$ among those $j$ with $M(j) = k$ with probability 
\[
	\frac{p_{j}}{\sum_{\{a | M(a) = k\}} p_a}.
\]
Thus, a trivial refinement does not yield any additional information
about the GPT-system. We call a measurement \emph{fine-grained} if it does not have any non-trivial
refinements. The set of fine-grained measurements on any state space $A$ is denoted $\mathcal{E}_A^*$.

Now we consider the R$\acute{\text{e}}$nyi entropies~\cite{Renyi}, which are defined for probability distributions $\mathbf{p}=(p_1,\ldots,p_n)$ as
\[
		H_\alpha(\mathbf p) = \frac 1 {1-\alpha} \log\left(\sum_j p_j^\alpha\right),
\]
	where $\alpha \in (0,\infty)\setminus\{1\}$. Furthermore,
\[
		H_0(\mathbf p) := \lim_{\alpha \rightarrow 0} H_\alpha(\mathbf p) = \log | \rm{supp}(\mathbf p)| 
\]
	with ${\rm{supp}}(\mathbf p)=\{p_j \ | \ p_j > 0\}$ is called the \emph{max-entropy}, and 
\[
		H_\infty(\mathbf p) := \lim_{\alpha \rightarrow \infty} H_\alpha(\mathbf p) = -\log \max_j p_j
\]
	is called the \emph{min-entropy}.
	Also, 
\[
		H_1(\mathbf p) := \lim_{\alpha \rightarrow 1} H_\alpha(\mathbf p) = -\sum_j p_j \log p_j = H(\mathbf p)
\]
	is just the regular Shannon entropy $H$.
	
	For $\alpha \in [0,\infty]$ and GPTs satisfying Postulates 1 and 2, we generalize the classical R$\acute{\text{e}}$nyi entropies:
\[
		S_\alpha(\omega) := H_\alpha(\mathbf p),
\]
	where $\omega=\sum_j p_j w_j$ is any decomposition into perfectly distinguishable pure states. According to Theorem~\ref{theorem:Observables}, the result is independent of the choice of decomposition. We have $S_1=S$, the spectral entropy of Definition~\ref{DefSpectralEntropy}.

Following~\cite{WehnerEntropy}, for every $\alpha \in [0,\infty]$ and $\omega\in\Omega_A$, we define the \emph{order-}$\alpha$ \emph{R$\acute{\text{e}}$nyi measurement entropy} as
\[
		\widehat{S}_\alpha(\omega) = \inf_{e \in \mathcal{E}_A^*} H_\alpha\left(\strut e_1(\omega),e_2(\omega),\ldots\right),
\]
	where $H_\alpha$ on the right-hand side denotes the classical R$\acute{\text{e}}$nyi entropy.
	The \emph{order-$\alpha$ R$\acute{\text{e}}$nyi decomposition entropy} is defined as
	\begin{equation}
		\check{S}_\alpha(\omega) := \inf_{\omega=\sum_j q_j \varphi_j} H_\alpha(\mathbf q),
	\end{equation}
	where the infimum is over all convex decompositions of $\omega$ into \emph{pure} states $\varphi_j\in\Omega_A$.
	
	The idea of measurement entropy is to characterize the state before a measurement. For
        example, in quantum theory, particles prepared in a state
        $\ket{\psi}$ which all give the same result in energy
        measurements would be said to be in an energy eigenstate. If
        instead we performed a position measurement, the resulting distribution of positions would have non-zero entropy. However, this entropy would arguably not come from the initial state, but from the measurement process itself due to the uncertainty principle.
        
Suppose we would like to prepare a state $\omega$ by using states of maximal knowledge (i.e.\ pure states) $\varphi_j$, and a random number generator which gives output $j$ with probability $p_j$. Then the decomposition entropy quantifies the smallest information content (entropy) of a random number generator that would be necessary to build such a device.
	For more detailed operational interpretations of measurement and decomposition entropy, in particular for $\alpha=1$, see~\cite{WehnerEntropy,BarnumEntropy} Note that in quantum theory, measurement, decomposition and spectral R\'enyi entropies all coincide, with the 
$\alpha=1$ case giving von Neumann entropy, $S(\omega)=-{\rm tr}(\omega\log\omega)$.
	
Our first result is that the spectral and measurement definitions of the entropies agree:
\begin{thm}
\label{TheRenyi}
	Consider any state space $A$ which satisfies Postulates 1 and 2. Then the R$\acute{\text{e}}$nyi entropies $S_\alpha$ and the R$\acute{\text{e}}$nyi measurement entropies $\widehat S_\alpha$ coincide, and upper-bound the R\'enyi decomposition entropy $\check S_\alpha$, i.e.
\[
		\check{S}_\alpha(\omega)\leq S_\alpha(\omega) = \widehat{S}_\alpha(\omega) \mbox{ for all } \omega \in \Omega_A,\ \alpha \in [0,\infty].
\]
	In particular, for $\alpha=1$, the measurement entropy $\widehat S$ is the same as the spectral entropy $S$ from Definition~\ref{DefSpectralEntropy}, which we have identified with thermodynamical entropy $H$ in Observation~\ref{ObsEntropy}.
	\label{theorem:MeasIsTherm}
\end{thm}
The inequality $\check{S}_\alpha\leq S_\alpha$ is easy to see: for a decomposition $\omega=\sum_i p_i\omega_i$ into perfectly distinguishable pure states $\omega_i$, the states $\omega_i$ can also be seen as a fine-grained measurement, yielding outcome probabilities $p_i$. So taking the infimum over all decompositions gives \emph{at most} $H_\alpha(\mathbf{p})=S_\alpha(\omega)$. The equality between $S_\alpha$ and $\widehat S_\alpha$ is shown in the appendix.

We do not know in general whether Postulates 1 and 2 imply that 
$\check{S}_\alpha=S_\alpha$ for all $\alpha$. Interestingly, we know it for $\alpha=2$ and $\alpha=\infty$:
\begin{thm} \label{thm:H2}
If a state space satisfies Postulates 1 and 2, then $\check{S}_2(\omega) = S_2(\omega)$ and $\check{S}_\infty(\omega) = S_\infty(\omega)$
	for all states $\omega$.
\end{thm}
\begin{proof}
To give the reader an idea of the kind of arguments involved, we present the proof for $S_2$, but defer the proof for $S_\infty$ to the appendix.
If $\omega=\sum_j p_j \omega_j$ is any convex decomposition into a maximal set of perfectly distinguishable pure states (without loss of generality $p_1\geq p_2\geq\ldots$), and $\omega=\sum_j q_j \varphi_j$ any (other) convex decomposition into pure states $\varphi_j$ (also with $q_1\geq q_2\geq\ldots$), then $\sum_j p_j^2 =\braket{\omega , \omega} = \sum_j q_j^2 + \sum_{j \ne k} q_j q_k \braket{\varphi_j , \varphi_k} \ge \sum_j q_j^2$ since $\langle \varphi_j,\varphi_k\rangle\geq 0$. Thus, we have
\[
	S_2(\omega)=-\log\sum_j p_j^2 \leq -\log \sum_j q_j^2 = H_2(\mathbf{q}),
\]
and since $\check{S}_2(\omega)$ is defined as the infimum over the right-hand side, we obtain that $\check{S}_2(\omega)\geq S_2(\omega)$; we find the converse inequality in Theorem~\ref{TheRenyi}.
\end{proof}

We do not know whether the same identity holds for the most interesting case $\alpha=1$, the case of standard thermodynamic entropy $S=S_1$. In the max-entropy case $\alpha=0$, however, we have a surprising relation to higher-order interference:

\begin{thm}
\label{ThmAllEquiv}
	Consider a state space satisfying Postulates 1 and 2. Then the following statements are all equivalent:
\begin{itemize}
\item[(i)] The state space does not have third-order interference.
\item[(ii)] The measurement and decomposition versions of max-entropy coincide, i.e.\ $\check{S}_0(\omega)=S_0(\omega)$ for all states $\omega$.
\item[(iii)] The state space is either classical, or one on the list of 
Theorem~\ref{LemList}.
\item[(iv)] If $\omega$ is a pure state and $P_F$ any orthogonal projection onto any face $F$, then $P_F\omega$ is a multiple of a pure state.
\item[(v)] The ``atomic covering property'' of quantum logic holds.
\end{itemize}
\end{thm}
The equivalences $(i)\Leftrightarrow(iii)\Leftrightarrow(iv)\Leftrightarrow(v)$ are shown in~\cite{Postulates}; our new result is the equivalence to (ii), which is shown in the appendix.

 Absence of third-order interference is meant in the
  sense of eq.~(\ref{eqNoThirdOrder}), as introduced originally by
  Sorkin~\cite{Sorkin}: only pairs of mutually exclusive alternatives
  can possibly interfere.  It is interesting that this is related to
  an information-theoretic property of max-entropy $S_0$, as given in
  (ii). We do not currently know whether $S_0$ (or, in particular, the identity of $\check S_0$ and $S_0$) has any thermodynamic relevance in the class of theories that we are considering, but it certainly does within quantum theory, where it attains operational meaning in single-shot thermodynamics~\cite{HorodeckiOppenheim,Gour}.
  
   As (iii) shows, this theorem is closely related to Open
  Problem 1: it gives properties of conceivable state
spaces that satisfy Postulates 1 and 2, but are not on the list of
known examples (namely, they do not satisfy any of $(i)-(v)$).
Similarly, (iv) shows the relation of higher-order interference to
Open Problem 2, and (v) relates all these items to quantum
logic. In fact, one can show that Postulates 1 and 2 imply that the set of faces of the state space has the structure of an \emph{orthomodular lattice}, which is often seen as the definition of quantum logic. For readers who are familiar with the terminology of quantum logic, we give some additional remarks in Subsection~\ref{AppendixQuantumLogic} in the appendix.

\section{Conclusions}
As discussed in the introduction, many works (dating back at least to the 1950s) have considered quantum theory as just one particular example of a probabilistic theory: a single point in a large space of theories that contains classical probability theory, as well as many other possibilities that are non-quantum and non-classical. More recent works have focused on the information-theoretic properties of quantum theory, for example deriving quantum theory as the unique structure that satisfies a number of information-theoretic postulates. 

Rather than attempt a derivation of quantum theory from postulates, this paper has examined the thermodynamic properties of quantum theory and of those theories that are similar enough to quantum theory to admit a good definition of thermodynamic entropy, and of some version of the Second Law. Postulate 1 states that there is a reversible transformation between any two sets of $n$ distinguishable pure states. This can be thought of as an expression of the universality of the representation of information, in particular that a choice of basis is arbitrary, and also allows for reversible microscopic dynamics, as is crucial for thermodynamics. Postulate 2 states that every state can be written as a convex mixture of perfectly distinguishable pure states. This ensures that a mixed state describing an ensemble of many particles can be treated as if each particle has an unknown microstate, drawn from a set of distinguishable possibilities.  

Much follows from Postulates 1 and 2, without needing to assume any other aspects of the standard formalism of quantum theory. In order to derive thermodynamic conclusions, we considered the argument originally employed by von Neumann in his derivation of the mathematical expression for the thermodynamic entropy of a quantum state. The argument involves a thought experiment with a gas of quantum particles in a box, and semi-permeable membranes that allow a particle to pass or not depending on the outcome of a quantum measurement. By applying the same thought experiment, we showed that given any theory satisfying Postulates 1 and 2, there is a unique expression for the the thermodynamic entropy, equal to both the spectral entropy and the measurement entropy. By way of contrast, a fictitious system defined by a square state space, which arises as Alice's local system of an entangled pair producing stronger-than-quantum ``PR box'' correlations, does not satisfy either Postulate. This system -- the gbit -- does not admit a sensible notion of thermodynamic entropy, at least not one that is given to it by the von Neumann or Petz arguments. While many works have discussed the inability of quantum theory to produce arbitrarily strong nonlocal correlations, this connection with thermodynamics deserves further investigation. It would be very interesting, for example, if Tsirelson's bound on the strength of quantum nonlocal correlations could be derived from a thermodynamic argument.  

There are many other consequences of Postulates 1 and 2 for both thermodynamic and information-theoretic entropies. For example, a form of the Second Law holds in that neither projective measurements nor mixing procedures can decrease the thermodynamic entropy. The spectral and measurement order-$\alpha$ Renyi entropies coincide for any $\alpha$. The spectral and decomposition order-$\alpha$ Renyi entropies coincide for $\alpha = 2$ or $\infty$. An open question is whether any theory satisfying Postulates 1 and 2 is completely satisfactory from the thermodynamic point of view. While the von Neumann and Petz arguments can be run with no trouble in the presence of Postulates 1 and 2 as we have shown, there could still be a different physical scenario, in which theories would fail to exhibit sensible behaviour unless they have even more of the structure of quantum theory.

Finally, another major open question is whether quantum-like theories exist, satisfying Postulates 1 and 2, that are distinct from quantum theory in that they admit higher-order interference. Roughly speaking, this means that three or more possibilities can interfere in order to produce an overall amplitude, unlike in quantum theory, where different possibilities only interfere in pairs. We extend the results of Ref.~\cite{Postulates}, where it was shown that in the context of Postulates 1 and 2 the existence of higher-order interference is equivalent to each of three other statements.  We provide an equivalent entropic condition: there is higher-order interference if and only if the measurement and decomposition versions of the max entropy do not coincide.

Our understanding of quantum theory would be greatly improved if higher-order interference could be ruled out by simple information-theoretic, thermodynamic, or other physical arguments. On the other hand, if theories with higher-order interference exist and are eminently sensible, an immediate question is whether an experimental test could be performed to distinguish such a theory from quantum theory. While previous experiments~\cite{Sinha,Sinha2009,Soellner,Kauten,Hickmann,Park} only tested for a zero versus non-zero value of higher-order interference, sensible higher-order theories that satisfy Postulates 1 and 2 (if they exist) could help to inform future experiments by supplying concrete models that can be tested against standard quantum theory.\\

\textbf{Acknowledgments.} We would like to thank
  Matt Leifer for many useful discussions, and we are grateful to the
  participants of the ``Foundations of Physics working group'' at
  Western University for helpful feedback. We would also like to thank Giulio Chiribella and Carlo Maria Scandolo for coordinating the arXiv posting of their work with us.
  This research was
supported in part by Perimeter Institute for Theoretical
Physics. Research at Perimeter Institute is supported by the
Government of Canada through the Department of Innovation, Science and
Economic Development Canada and by the Province of Ontario through the
Ministry of Research, Innovation and Science. This research was
undertaken, in part, thanks to funding from the Canada Research Chairs
program. This research was supported by the FQXi
  Large Grant ``Thermodynamic vs.\ information theoretic entropies in
  probabilistic theories''. HB thanks the Riemann Center for
  Geometry and Physics at the Institute for Theoretical Physics,
  Leibniz University Hannover, for support as a visiting fellow during part of the time this paper was in preparation.

\onecolumngrid

\section{Appendix}

\subsection{Proofs}

\subsubsection{Proof that observables are well-defined}
In this appendix, a decomposition of a state into perfectly distinguishable pure states (which always exists due to Postulate 2) will be called a ``classical decomposition''.
\begin{lem} \label{Lemma:DecompositionInFace}
	Assume Postulates 1 and 2. Let $F \ne \{0\}$ be a face of $A_+$ and $\omega \in \Omega_A \cap F$. Then there exists a classical decomposition $\omega = \sum_j p_j \omega_j$ with $\omega_j \in F$ for all $j$.
\end{lem}

\begin{proof}
	Let $\omega = \sum_j p_j \omega_j$ be a classical decomposition with $p_j \ne 0$. As $\omega \in F$ and $F$ a face, $\omega_j \in F$ for all $j$.
\end{proof}

\begin{proof}[Proof of Theorem \ref{theorem:Observables}]
	Let $x \in A$ be arbitrary. By Lemma 5.46 from \cite{UdudecDoktor} there exists a frame $\{\omega_j\}$ and $x'_j \in \mathbb R$ such that $x = \sum_j x'_j \omega_j$. We extend $\{\omega_j\}$ to a maximal frame by adding $x'_j := 0$ for the new indices $j$. Now we group together the $j$ with the same $x_j'$ value, and by relabelling we find that $x = \sum_{k=1}^n x_k \sum_i \omega_{k;i}$ where the $x_k$ are pairwise different values of the $x_j'$ and the $\omega_{k;i}$ are the $\omega_j$ that belong to this $x_j'$ value. For any given $k$, the $\omega_{k;i}$ generate a face $F_k$ with projective unit $u_k = \sum_i \omega_{k;i}$.
	
	Therefore we find a decomposition $x = \sum_{k=1}^n x_k u_k$ with $x_k$ pairwise different real numbers and $u_k$ order units of faces $F_k$ and $\sum_{k=1}^n u_k = u_A$.\\
	
	Now we show that the faces $F_k$ are mutually orthogonal:\\
	Let $\omega \in F_k$ be an arbitrary normalized state. By Lemma \ref{Lemma:DecompositionInFace} it has a classical decomposition $\omega = \sum_j p_j \omega_j^{(k)}$ which uses only pure states $\omega_j^{(k)} \in F_k$. Wlog we assume that these pure states form a generating frame of $F_k$, by extending the frame and adding $p_j=0$ to the decomposition. Consider another face $F_m$, i.e. $m\ne k$. Likewise to $\omega$, let $\omega' \in F_m$ be an arbitrary normalized state and $\omega' = \sum_j q_j \omega_{j}^{(m)}$ be a classical decomposition with $\omega_{j}^{(m)}$ a generating frame for $F_m$. For the other faces define $\omega_j^{(i)} := \omega_{i;j}$. Then $u_i = \sum_j \omega_j^{(i)}$ and in total $u_A = \sum_i u_i = \sum_i \sum_j \omega_{j}^{(i)}$. As $\braket{\nu,\nu}=1$ for all pure states $\nu \in \Omega_A$, this implies that the $\omega_{j}^{(i)}$ are mutually orthogonal: 
\[
		1 = u_A(\omega_h^{(g)}) = \sum_i \sum_j \braket{\omega_{j}^{(i)},\omega_h^{(g)}} = 1 + \sum_{(i,j) \ne (g,h)} \braket{\omega_{j}^{(i)},\omega_h^{(g)}}
\]
	and therefore $\braket{\omega_{j}^{(i)},\omega_h^{(g)}} \ge 0$ implies $\braket{\omega_{j}^{(i)},\omega_h^{(g)}} = 0$ for all $(i,j) \ne (g,h)$. Thus we find $\braket{\omega, \omega'} = \sum_j \sum_b p_j q_b \braket{\omega_j^{(k)},\omega_b^{(m)}} = 0$ because $m \ne k$. As $\omega \in F_k$ and $\omega' \in F_m$ were arbitrary (normalized) states, this implies that $F_k$ and $F_m$ are orthogonal. As $k\ne m$ were arbitrary, all the faces are mutually orthogonal.
	\\
	
	Now we will show that the decomposition $x= \sum_j x_j u_j$ is unique. So assume there are two decompositions $x=\sum_{j=1}^{n_a} a_j u_{j}^{(a)} = \sum_{j=1}^{n_b} b_j u_{j}^{(b)}$ with $a_j \in \mathbb R$ pairwise different and projective units $u_{j}^{(a)}$ that add up to the order unit (analogously for $b$) and belong to pairwise orthogonal faces $F_j^{(a)}$. Wlog we assume that the $a_j$ and $b_j$ are ordered by size, i.e. $a_1<a_2<...<a_{n_a}$. We want to show $a_1 = b_1$. The $u_{j}^{(a)}$ generate the faces $F_j^{(a)}$. Let $\omega_{j;i}^{(a)}$ be a generating frame for the face $F_j^{(a)}$, especially $\sum_i \omega_{j;i}^{(a)} = u_j^{(a)}$. As the faces are mutually orthogonal and the projective units add up to $u_A$, the $\omega_{j;i}^{(a)}$ form a maximal frame; in particular they add up to $u_A$ (likewise for $b$). Therefore:
\[
		a_1 = \braket{\omega_{1;j}^{(a)}, x} = \sum_{k,i} b_k \braket{\omega_{1;j}^{(a)}, \omega_{k;i}^{(b)}} \ge \sum_{k,i} b_1 \braket{\omega_{1;j}^{(a)}, \omega_{k;i}^{(b)}} = b_1 u_A(\omega_{1;j}^{(a)}) = b_1.
\]
	Analogously show $b_1 \ge a_1$, i.e. $b_1 = a_1$ in total.
	
	Now suppose there was a $k>1$ and an $i$ with $\braket{\omega_{1;j}^{(a)},\omega_{k;i}^{(b)}} \ne 0$, i.e. $\braket{\omega_{1;j}^{(a)},\omega_{k;i}^{(b)}} > 0$. Then
	\begin{align*}
		a_1 &= \braket{\omega_{1;j}^{(a)},x} = \sum_{k,i} b_k \braket{\omega_{1;j}^{(a)},\omega_{k;i}^{(b)}} = a_1 \sum_i \braket{\omega_{1;j}^{(a)},\omega_{1;i}^{(b)}} + \sum_{k > 1,i} b_k \braket{\omega_{1;j}^{(a)},\omega_{k;i}^{(b)}}
		> a_1 \sum_i \braket{\omega_{1;j}^{(a)},\omega_{1;i}^{(b)}} + \sum_{k > 1,i} a_1 \braket{\omega_{1;j}^{(a)},\omega_{k;i}^{(b)}} \\
		&= a_1 \sum_{k,i} \braket{\omega_{1;j}^{(a)},\omega_{k;i}^{(b)}} = a_1 u_A(\omega_{1;j}^{(a)}) = a_1.
	\end{align*}
	This is a contradiction. Thus $\braket{\omega_{1;j}^{(a)},\omega_{k;i}^{(b)}} = 0$ for all $k > 1$ and $i$.
	Therefore we find $u_1^{(b)}(\omega_{1;i}^{(a)}) = \sum_j \langle \omega_{1;j}^{(b)},\omega_{1;i}^{(a)}\rangle=\sum_{j,k}\langle \omega_{k,j}^{(b)},\omega_{1;i}^{(a)}\rangle=u_A(\omega_{1;i}^{(a)}) = 1$ and analogously $u_1^{(a)}(\omega_{1;i}^{(b)})=1$. By Proposition 5.29 from \cite{UdudecDoktor}, we have $\Omega_A \cap F = \{ \omega \in \Omega_A \ | \ u_F(\omega) = 1 \}$. Therefore a generating frame of $F_1^{(a)}$ is contained in $F_1^{(b)}$ and vice versa. Thus we find $F_1^{(a)} = F_1^{(b)}$ and $u_1^{(a)} = u_1^{(b)}$.
	
For the remaining indices, we construct an inductive proof: Choose $L\in\mathbb{R}$ large enough such that $a_1+L>\max\{a_{n_a},b_{n_b}\}$, and define $x' := x + L\cdot u_1^{(a)}$, i.e. $x' = \sum_{j=1}^{n_a} (a_j+\delta_{j,1}\cdot L) u_j^{(a)} = \sum_{j=1}^{n_b} (b_j+\delta_{j,1}\cdot L) u_j^{(b)}$. Furthermore defining $a_1' := a_2$, $a_2' := a_3$,..., $a_{n_a}' := a_1 + L$, $u_1^{(a')} := u_2^{(a)}$, $u_2^{(a')} := u_3^{(a)}$,...,$u_{n_a}^{(a')} := u_1^{(a)}$ and likewise for $b_j'$, we find $x' = \sum_{j=1}^{n_a} a_j' u_j^{(a')} = \sum_{j=1}^{n_b} b_j' u_j^{(b')}$ with $a_1' < a_2' <...<a_{n_a}'$ and $b_1' < b_2' <...<b_{n_b}'$. Repeating the exact same procedure as before, we obtain $a_1' = b_1'$ and $u_1^{(a')} = u_1^{(b')}$, i.e. $a_2=b_2$ and $u_2^{(a)} = u_2^{(b)}$. We iterate to find $a_j = b_j$ and $u_j^{(a)} = u_j^{(b)}$ for all $j$. Note that as all maximal frames have the same size and as the projective units add up to $u_A$, necessarily $n_a = n_b$.\\
	
	At last we construct the projective measurement that corresponds to measuring the observable $x$:
	
	For $F_k$, let $P_k$ be the orthogonal projector onto the span of $F_k$ (in particular, $P_k : A \rightarrow \mathrm{span}(F_k)$ surjective). We know that these projectors are positive and linear and satisfy $u_A \circ P_k = u_k$. Furthermore $0\le u_k = u_A \circ P_k \le u_A$ and $\sum_k u_A \circ P_k = \sum_k u_k = u_A$, i.e. we obtain a well-defined measurement; therefore the $P_k$ form a well-defined instrument. As they are projectors, the $P_k$ leave the elements of $F_k$ unchanged.
\end{proof}

\subsubsection{Proof of Observation~\ref{ObsEntropy}}
In order to show that $H(\omega)=S(\omega)$ is consistent with Assumptions~\ref{AssThoughtExp}, we only have to show that $\omega\mapsto S(\omega)$ is continuous, to comply with assumption (d). According to Theorem~\ref{TheRenyi} (which we will prove below), the spectral entropy $S(\omega)$ equals measurement entropy $\widehat S(\omega)$. But it is well-known~\cite{BarnumEntropy} and easy to see from its definition that $\widehat S$ is continuous.

It remains to show eq.~(\ref{eqSDecomp}). So let $\omega=\sum_j p_j \omega_j$ be any decomposition of $\omega$ into perfectly distinguishable, not necessarily pure states $\omega_i$. Decompose all the $\omega_i$ into perfectly distinguishable pure states $\omega_j^{(i)}$, i.e.\ $\omega_j=\sum_i q_j^{(i)} \omega_j^{(i)}$. Perfectly distinguishable states live in orthogonal faces, thus $\langle \omega_i,\omega_j\rangle=0$ for $i\neq j$ (note that this is a conclusion that follows from Postulates 1 and 2, but could not be drawn from bit symmetry alone in~\cite{Skalarprodukt}). Thus, we also have $\langle \omega_i^{(j)},\omega_k^{(l)}\rangle=0$ for $i\neq k$ or $j\neq l$, and so $\omega=\sum_{ij} p_j q_j^{(i)} \omega_j^{(i)}$ is a decomposition of $\omega$ into perfectly distinguishable pure states. Define the real function $\eta:[0,1]\to\mathbb{R}$ via $\eta(x):=-x\log x$ for $x>0$ and $\eta(0)=0$. Due to Theorem~\ref{theorem:Observables} and $\eta(xy)=-xy\log x -xy \log y$, we have
\[
   \eta(\omega)=\sum_{ij} \eta(p_j q_j^{(i)})\omega_j^{(i)} = -\sum_{ij} p_j q_j^{(i)} \log p_j \omega_j^{(i)} - \sum_{ij} p_j q_j^{(i)} \left(\log q_j^{(i)}\right)\omega_j^{(i)},
\]
and therefore
\[
   S(\omega)=u_A(\eta(\omega)) = 
   -\sum_{ij} p_j q_j^{(i)} \log p_j - \sum_{ij} p_j q_j^{(i)} \log q_j^{(i)} = -\sum_j p_j \log p_j + \sum_j p_j\underbrace{\left(-\sum_i q_j^{(i)} \log q_j^{(i)}\right)}_{S(\omega_j)}
\]
This completes the proof of Observation~\ref{ObsEntropy}.
\qed

\subsubsection{Proof of the second half of Theorem~\ref{thm:H2}}
Use the notation of the first half of the proof. We claim that $\max_{\varphi\in\Omega}\langle\omega,\varphi\rangle=p_1$. The inequality ``$\geq$'' is trivial (consider the special case $\varphi=\omega_1$). To see the inequality ``$\leq$'', note that $\langle \omega,\varphi\rangle = \sum_j p_j \lambda_j$, where $\lambda_j:=\langle \omega_j,\varphi\rangle\in [0,1]$ satisfies $\sum_j \lambda_j = \langle \sum_j \omega_j,\varphi\rangle=\langle u,\varphi\rangle=1$, and so $\langle \omega,\varphi\rangle\leq p_1$ for all $\varphi$. Thus
\[
S_\infty(\omega)=-\log p_1 = -\log \max_{\varphi\in\Omega} \langle \omega,\varphi\rangle \leq -\log \langle\omega,\varphi_1\rangle
=-\log\left(q_1+\sum_{j\geq 2} q_j \langle\varphi_j,\varphi_1\rangle\right) \leq -\log q_1
= H_\infty(\mathbf{q}).
\]
Similarly as in the first part of the proof, we obtain $\check{S}_\infty(\omega)\geq S_\infty(\omega)$.
The converse inequality from Theorem~\ref{TheRenyi} for $\alpha=\infty$ concludes the proof.
\qed

\subsubsection{Proof of Klein's Inequality and the Second Law for projective measurements}
We consider an ensemble of systems described by an arbitrary state $\omega \in \Omega_A$. To all systems of this ensemble we apply a projective measurement described by orthogonal projectors $P_a$ which form an instrument, resulting in a new ensemble state $\omega'$. The $P_a$ project onto the linear span of faces $F_a$ that replace the eigenspaces from quantum theory. We want to show that the measurement cannot decrease the entropy of the ensemble, i.e.
\[
	S(\omega') \ge S(\omega).
\]
We decompose the proof into several steps. Our basic idea follows the proof of a similar statement for quantum theory in \cite{NielsenChuang}: We reduce the proof of the Second Law to Klein's inequality. But as we do not have access to an underlying pure state Hilbert space, we will need to use a different argument for why Klein's inequality implies the Second Law for projective measurements.

So at first we prove Klein's inequality, adapting the proof of \cite{NielsenChuang}. We note that a similar proof has also been found by Scandolo~\cite{ScandoloMaster}, albeit under different assumptions.
\begin{proof}[Proof of Theorem \ref{theorem:Klein}]
	We consider two arbitrary states $\omega, \nu$ with classical decompositions $\omega = \sum_j p_j \omega_j$, $\nu = \sum_k q_k \nu_k$, where wlog the $\omega_j$ and the $\nu_k$ form maximal frames. 
	We define the matrix $P_{jk} := \braket{\omega_j , \nu_k}$. All its components are non-negative, i.e. $P_{jk} \ge 0$, because the scalar product itself is non-negative for all states. As all maximal frames have the same size, the matrix is a square matrix; as maximal frames sum to $u_A$, the rows and columns sum to one: $\sum_j P_{jk} = \sum_k P_{jk} = 1$. Thus, we get
\[
		S(\omega \| \nu) = -S(\omega) - \braket{\omega , \log \nu} =\sum_j p_j \log p_j - \sum_{jk} p_j \log q_k \braket{\omega_j, \nu_k} = \sum_j p_j \left( \log p_j - \sum_k P_{jk} \log q_k \right).
\]
	We define $r_j := \sum_k P_{jk} q_k$. Note that the $r_j$ form a probability distribution: $r_j \ge 0$ and $\sum_j r_j = \sum_k \sum_j P_{jk} q_k = \sum_k q_k = 1$. Using the strict concavity of the logarithm, we find:
\[
		\log r_j = \log\left(\sum_k P_{jk} q_k\right) \ge \sum_k P_{jk} \log q_k.
\]
	Therefore we get
\[
		S(\omega \| \nu) = \sum_j p_j \left( \log p_j - \sum_k P_{jk} \log q_k \right) \ge \sum_j p_j \left( \log p_j - \log r_j  \right) = \sum_j p_j \log \left( \frac{p_j}{r_j} \right).
\]
	We recognize the last expression as the classical relative entropy of the probability distributions $p_j$ and $r_j$. This classical relative entropy has the important property that it is never negative. Thus:
\[
		S(\omega \| \nu) \ge 0.
\]
\end{proof}

In order to get the main proof less convoluted, we will state some technical parts as lemmas.

\begin{lem} \label{Lemma:ProjectorsOrthogonal}
	Assume Postulate 1 and 2. Consider orthogonal projectors $P_j$ which form an instrument. Then the $P_j$ are mutually orthogonal:
\[
		P_k P_j = \delta_{jk} P_j.
\]
\end{lem}
\begin{proof}
	We prove $P_k P_j \omega = 0$ for all $\omega \in A$, $j \ne k$. If $P_j \omega = 0$ this is trivial, so from now on assume $P_j \omega \ne 0$. As the cone is generating  (i.e. $\mathrm{Span}(A_+) = A$) and the projectors linear, it is sufficient to show $P_k P_j \omega = 0$ for all $w \in A_+$.\\
	As $P_j$ is positive, $P_j \omega \ne 0$ implies that $(u_A \circ P_j)(\omega) > 0$ because only the zero-state is normalized to $0$.\\
	Using $u_A = u_A \circ (\sum_j P_j) = \sum_j u_A \circ P_j$ and $P_j P_j= P_j$:
	\begin{align*}
		1 &= u_A\left( \frac{P_j \omega}{(u_A \circ P_j)(\omega)} \right) = \left(u_A \circ \sum_k P_k \right)\left( \frac{P_j \omega}{(u_A \circ P_j)(\omega)} \right) =  u_A\left( \frac{ P_j P_j \omega}{(u_A \circ P_j)(\omega)} \right) + u_A \left( \frac{\sum_{k | k \ne j} P_k P_j \omega}{(u_A \circ P_j)(\omega)} \right)\\
		&= u_A\left( \frac{ P_j \omega}{(u_A \circ P_j)(\omega)} \right) + u_A \left( \frac{\sum_{k | k \ne j} P_k P_j \omega}{(u_A \circ P_j)(\omega)} \right) \qquad \Rightarrow \qquad 0 = u_A \left( \frac{\sum_{k | k \ne j} P_k P_j \omega}{(u_A \circ P_j)(\omega)} \right) = \sum_{k | k \ne j} u_A \left( P_k P_j \omega \right).
	\end{align*}
	As the projectors are positive and only the zero-state is normalized to $0$, this shows $P_k P_j \omega = 0$ for $k\ne j$.
\end{proof}

\begin{lem}
	Assume Postulates 1 and 2. Consider an orthogonal projector $P$ which projects onto the linear span of a face $F$ of $A_+$. Then for all states $\omega \in A_+$ we find $P\omega \in F$.
\end{lem}
\begin{proof}
From basic convex geometry (see e.g.\ Proposition 2.10 in \cite{PfisterMaster}), we know that $F={\rm span}(F)\cap A_+$. Since $P$ is positive, we have $P\omega\in A_+$; furthermore, since $P$ projects onto $F$, we have $P\omega\in {\rm span}(F)$, thus $P\omega\in F$.
\end{proof}

\begin{proof}[Proof of Theorem \ref{theorem:SecondLaw}]
	We know that $S(\omega\|\omega') = -S(\omega) - \braket{\omega, \log \omega'} \ge 0$. As in Theorem 11.9 from \cite{NielsenChuang} , we claim $-\braket{\omega, \log \omega'} = S(\omega')$ and therefore $-S(\omega) + S(\omega') \ge 0$. Thus we only have to prove $-\braket{\omega, \log \omega'} = S(\omega')$. But as we do not have access to an underlying pure state Hilbert space, our proof is different from~\cite{NielsenChuang}.

	By Lemma \ref{Lemma:ProjectorsOrthogonal}, the $P_a$ are mutually orthogonal, i.e. $P_a P_b = \delta_{ab} P_b$. By symmetry of the $P_a$ also the $P_a \omega$ are mutually orthogonal: $\braket{P_a \omega, P_b \omega} = \braket{\omega, P_a P_b \omega} = 0$ for $a \ne b$. This also shows that the $F_a$ are mutually orthogonal. If $P_a \omega = 0$ we use the decomposition $P_a \omega = u_A(P_a \omega) \sum_k r_{ak} w_{ak}$ with $r_{ak} = \delta_{ak}$ and $w_{ak}$ an arbitrary generating frame of $F_a$. If $P_a \omega \ne 0$, then $\frac{P_a \omega}{u_A(P_a \omega)} \in F_a \cap \Omega_A$ and by Lemma \ref{Lemma:DecompositionInFace}, there is a classical decomposition $\frac{P_a \omega}{u_A(P_a \omega)} = \sum_k r_{ak} \omega_{ak}$ with $\omega_{ak} \in F_a$. We complete the $\omega_{ak}$ to generating frames of the $F_a$ by adding terms with $r_{ak} = 0$. As we are using classical decompositions/frames, we know $\braket{\omega_{aj},\omega_{ak}} = \delta_{jk}$. Furthermore, as the $F_a$ are mutually orthogonal, we know $\braket{\omega_{aj}, \omega_{bk}} = 0$ for $a \ne b$.
	
	We note that the the $w_{aj}$ form a maximal frame:
\[
		u_A = \sum_a u_A \circ P_a  = \sum_a u_{F_a} = \sum_a \sum_j \omega_{aj}.
\]
For $a \ne b$ we have $P_b \omega_{aj} = P_b P_a \omega_{aj} = 0$, so we have a classical decomposition
\[
		\omega' = \sum_a P_a \omega = \sum_a \sum_j u_A(P_a \omega) r_{aj} \omega_{aj} 
\]
	with $\omega_{aj}$ a maximal frame that satisfies $P_a \omega_{bj} = \delta_{ab} \omega_{bj}$.	Note that we do not need to normalize $\omega'$ as the measurement itself is required to be normalized. Using
\[
		\sum_a P_a \log \omega' = \sum_{bj} \log(u_A(P_b \omega) r_{bj}) \sum_a P_a \omega_{bj} = \sum_{bj} \log(u_A(P_b \omega) r_{bj}) \omega_{bj} = \log \omega'
\]
	and
	\begin{equation}
		-S(\omega') = - \sum_{bj} (u_A(P_b \omega) r_{bj}) \log(u_A(P_b \omega) r_{bj}) = \braket{\omega', \log \omega'}
	\end{equation}
	as well as the symmetry of the $P_a$ we finally find:
	\begin{equation}
		-S(\omega') = \braket{\omega', \log \omega'} = \braket{\sum_a P_a \omega, \log \omega'} = \braket{\omega, \sum_a P_a \log \omega'} = \braket{\omega , \log \omega'}.
	\end{equation}
\end{proof}

\subsubsection{Proof that measurement and spectral entropies are identical}
In the main text we encountered different ways to define the entropy. One of them is to adapt classical entropy definitions by using the coefficients of a classical decomposition. Another is to adapt classical entropy definitions by using measurement probabilities and minimizing over all fine-grained measurements. Here we will show that in the context of Postulates 1 and 2, these two concepts yield the same R\'{e}nyi entropies.

To prove this, we will first analyze fine-grained measurements in further detail. The results will allow us to reproduce the quantum proof found in~\cite{WehnerEntropy} for our GPTs.

\begin{lem} \label{Lemma:FineGrainedDetails}
	Assume Postulates 1 and 2. Consider an arbitrary fine-grained measurement $(e_1,...,e_n)$. Then for all $j$ there exist some $c_j \in [0,1]$ and a pure state $\omega_j \in \Omega_A$ such that $e_j = c_j \braket{\omega_j, \cdot}$.
\end{lem}

\begin{proof}
	If $e_j = 0$, we can just take $c_j = 0$ and any pure state $\omega_j$. So from now on assume $e_j \ne 0$.
	
	Because of self-duality there exists some $\omega' \in A_+$ such that $\braket{\omega' , \cdot} = e_j$. As $e_j \ne 0$ also $\omega' \ne 0$ and therefore $u_A(\omega') \ne 0$. With $A_+ = \mathbb R_{\ge 0} \cdot \Omega_A$ and $c_j := u_A(\omega') > 0$ there exists an $\omega \in \Omega_A$ such that $\omega' = c_j \cdot \omega$. We want to prove that $\omega$ is pure, so assume it was not pure. Then it has a classical decomposition $\omega = \sum_{k=0}^N p_k \omega_k$ with $p_k > 0$ and $N \ge 1$. By relabelling we can assume $j = n$, i.e. we consider $e_n = c_j \sum_{k=0}^N p_k \braket{\omega_k , \cdot}$. Define a measurement $(e_1',...,e_{n+N}')$ by $e_k' := e_k$ for all $k = 1,2,...,n-1$ and $e_{n+i} ':= c_j p_i \braket{\omega_i,\cdot}$ for all $i = 0,1,...,N$. Because of $0\leq c_j p_i \braket{\omega_i,\cdot} = e_{n+i}'$ and $\sum_{k=1}^{n+N} e_k' = \sum_{k=1}^{n-1} e_k + \sum_{i=0}^N c_j p_i \braket{\omega_i,\cdot} = \sum_{k=1}^n e_k = u_A$ this is a well-defined measurement.
	
	Now define $M : \{1,...,n+N\} \rightarrow \{1,...,n\}$ by $M(i) := i$ for all $i = 1,...,n-1$ and $M(i) := n$ for all $i=n,...,n+N$. Then we get
	\begin{align*}
		\sum_{ \{a | M(a) = i\}} e_a' &= e_i & \text{ for } i < n  \\	
		\sum_{ \{a | M(a) = i\}} e_a' &= \sum_{a=n}^{n+N} e_i' = e_n & \text{ for } i = n.
	\end{align*}
	Thus the measurement $(e_1',...,e_{n+N}')$ is a refinement of $(e_1,...,e_n)$. With $e_n'(\omega_0) = c_j p_0 = e_n(\omega_0)$ and $ e_n'(\omega_1) = 0 \ne e_n(\omega_1)$ we find that $e_n'$ is not proportional to $e_n$, thus the fine-graining is non-trivial. This is in contradiction to our assumptions. Thus $\omega$ has to be pure. Furthermore $1= u_A(\omega) \ge e_j(\omega) = c_j \braket{\omega,\omega} = c_j$.
	
	So in total we have found that $e_j = c_j \braket{\omega,\cdot}$ with $\omega \in \Omega_A$ pure and $c_j \in [0,1]$. 
\end{proof}

\begin{lem} \label{Lemma:PerfectlyDistinguishingAndFineGrained}
	Assume Postulates 1 and 2. Let $\omega \in \Omega_A$ and $\omega = \sum_{j=1}^d p_j \omega_j$ be a decomposition into a maximal frame. Then the measurement that perfectly distinguishes the $\omega_j$ (i.e.\ $e_k(\omega_j) = \delta_{jk}$) can be chosen to be fine-grained.
\end{lem}

\begin{proof}
	Define $e_j := \braket{\omega_j, \cdot}$. As maximal frames add up to the order unit, this is a well-defined measurement and it satisfies $e_j(\omega_k) = \delta_{jk}$. It remains to show that this measurement is fine-grained.
	
	Consider a fine-graining $e_k'$ with $e_i = \sum_{\{ j | M(j) = i \}} e_j'$. By self-duality, there exist $c_j \ge 0$ and $\omega_j' \in \Omega_A$ such that $e_j' = c_j \braket{\omega_j',\cdot}$ and therefore $\sum_{\{ j | M(j) = k \}}c_j \omega_j' = \omega_k$. As $1 = u_A(\omega_k) = \sum_{\{ j | M(j) = k \}}c_j u_A(\omega_j') = \sum_{\{ j | M(j) = k \}}c_j$ we find that $\sum_{\{ j | M(j) = k \}}c_j \omega_j' = \omega_k$ is a convex decomposition of a pure state. This requires $c_j = 0$  or $\omega_j' = \omega_k$. In both cases $e_j' = c_j \braket{\omega_k , \cdot} = c_j e_k$ holds true for all $j$ with $M(j) = k$. Therefore, the fine-graining is trivial.
\end{proof}

\begin{lem} \label{Lemma:Bistochastic}
	Assume Postulates 1 and 2. Consider a fine-grained measurement $\mathbf e = (e_1,...,e_N) \in \mathcal E^*$. Then the maximal number of perfectly distinguishable states $d$ (often denoted as $N_A$) satisfies $d \le N$.
	
	Furthermore, consider a state $\omega \in \Omega_A$ with classical decomposition $\omega = \sum_{j=1}^d p_j \omega_j$ into a maximal frame. 
	Define the vector $\mathbf q := (e_j(\omega))_{1\le j \le N}$ of outcome probabilities and the $N$-component vector $\mathbf p = (p_1,...,p_d,0,...,0)\in\mathbb{R}^N$. Then $\mathbf q \prec \mathbf p$, i.e. there exists a bistochastic $N\times N$-matrix $M$ with $\mathbf q = M \mathbf p$. 
\end{lem}

\begin{proof}
	By Lemma \ref{Lemma:FineGrainedDetails} there exist $c_j \in [0,1]$ and pure $\omega_j' \in \Omega_A$ such that $e_j = c_j \braket{\omega_j',\cdot}$. Define $q_l := e_l(\omega) = c_l \braket{\omega_l',\omega}$. Using $\sum_{j=1}^N e_j = u_A$ and $\sum_{j=1}^d \nu_j = u_A$ for an arbitrary maximal frame $(\nu_1,...,\nu_d)$ we find:
\[
		\sum_{j=1}^N c_j = \sum_{j=1}^N c_j u_A(\omega_j') = u_A\left(\sum_{j=1}^N c_j \omega_j'\right) = u_A(u_A) = u_A\left(\sum_{j=1}^d \nu_j\right) = \sum_{j=1}^d u_A(\nu_j) = \sum_{j=1}^d 1 = d.
\]
	As $c_j \in [0,1]$, $\sum_{j=1}^N c_j = d$ shows $d \le N$.
	
	Set $q_{l|j} := e_l(\omega_j)$, introduce the $N$-component vector $\mathbf p := (p_1,...,p_d,0,...,0)$ and use that measurement effects and states of maximal frames add up to the order unit:
	\begin{align*}
		\sum_{j=1}^d q_{l|j} p_j &= \sum_{j=1}^d e_l(p_j \omega_j) = e_l(\omega) = q_l,\\
		\sum_{j=1}^d q_{l|j} &= \sum_{j=1}^d e_l(\omega_j) = c_l \sum_{j=1}^d \braket{\omega_l',\omega_j} = c_l u_A(\omega_l') = c_l,\\
		\sum_{l=1}^N q_{l|j} &= \sum_{l=1}^N e_l(\omega_j) = u_A(\omega_j) = 1.
	\end{align*}
	For $j \le d$ we define $M_{l,j} := q_{l|j}$. If $d<N$ we also define $M_{l,j} := \frac{1-c_l}{N-d}$ for $N \ge j > d$. M is an $N \times N$-matrix and it is bistochastic: first of all, $M_{l,j} \ge 0$ for all $l,j$. Furthermore:
	\begin{align*}
		\sum_{l=1}^N M_{l,j} &= \sum_{l=1}^N q_{l|j} = 1 &\text{ for } j \le d,\\
		\sum_{l=1}^N M_{l,j} &= \sum_{l=1}^N \frac{1-c_l}{N-d} = \frac{N-d}{N-d} = 1 &\text{ for } j > d,\\
		\sum_{j=1}^N M_{l,j} &= \sum_{j=1}^d q_{l|j} + (N-d)\cdot \frac{1-c_l}{N-d} = c_l + 1 -c_l = 1.
	\end{align*}
	This bistochastic matrix maps $\mathbf p$ to $\mathbf q$, i.e. $M \cdot \mathbf p = \mathbf q$:
\[
		\sum_{j=1}^N M_{l,j} p_j = \sum_{j=1}^d q_{l|j}  p_j = q_l.
\]
\end{proof}

	Now we come to the proof of the theorem:	
\begin{proof}[Proof of Theorem \ref{theorem:MeasIsTherm}]
	Consider an arbitrary fine-grained measurement $(e_1,...,e_N)$ and an arbitrary state $\omega \in \Omega_A$ with classical decomposition $\omega = \sum_{j=1}^d p_j \omega_j$ into a maximal frame. Define $q_l := e_l(\omega)$ and the $N$-component vector $\mathbf p = (p_1,...,p_d,0,...,0)$. Let $M$ be the bistochastic matrix from Lemma \ref{Lemma:Bistochastic} with $\mathbf q = M \cdot \mathbf p$. By Birkhoff's theorem, it is a convex combination of permutation matrices, i.e. $M = \sum_{\sigma \in S_N} a_\sigma P_\sigma$ for a probability distribution $a_\sigma$ and permutation matrices $P_\sigma$. Wlog we only consider the Shannon entropy; the proof for the R\'{e}nyi entropies works exactly the same way. As the Shannon entropy is Schur-concave and invariant under permutations:
\[
		H(\mathbf q) \ge \sum_{\sigma \in S_N} a_\sigma H(P_\sigma \cdot \mathbf p) = \sum_{\sigma \in S_N} a_\sigma H(\mathbf p) = H(\mathbf p) = S(\omega).
\]
	Furthermore $H(\mathbf p) = -\sum_{j=1}^d p_j \log p_j = S(\omega)$ is the entropy of a measurement that perfectly distinguishes the $\omega_j$, i.e. $e_j(\omega_k) = \delta_{jk}$. Because of Lemma \ref{Lemma:PerfectlyDistinguishingAndFineGrained}, such a measurement can be chosen to be finegrained. Therefore we find:
\[
		\widehat{H}(\omega) = \inf_{\mathbf e \in \mathcal E^*} H(\mathbf e(\omega)) = H(\mathbf p) = S(\omega).
\]
\end{proof}

\subsubsection{Proof of Theorem~\ref{ThmAllEquiv}}
As mentioned in the main text, the equivalences $(i)\Leftrightarrow(iii)\Leftrightarrow(iv)\Leftrightarrow(v)$ are shown in~\cite{Postulates}. We will now prove the equivalence $(ii)\Leftrightarrow(v)$, which proves Theorem~\ref{ThmAllEquiv}. Taking into account Theorem~\ref{TheRenyi}, and formulating the atomic covering property in the context of theories that satisfy Postulates 1 and 2, it remains to show the equivalence of the following two statements:
\begin{itemize}
\item[(ii')] For all states $\omega\in\Omega_A$, we have $\check S_0(\omega)\geq S_0(\omega)$.
\item[(v')] If $F$ is any face of $A_+$, and $\omega$ is any pure state, then the smallest face $G$ that contains both $F$ and $\omega$ has rank $|G|\leq |F|+1$. (Note that this is trivial if $\omega\perp F$.)
\end{itemize}

We will first prove that $(ii')\Rightarrow(v')$, which is equivalent to $\neg(v')\Rightarrow\neg(ii')$. So suppose that there exists some face $F$ of $A_+$ and a pure state $\omega$ such that the face $G$ generated by both has rank $|G|\geq |F|+2$. Let $\omega_1,\ldots,\omega_{|F|}$ be a frame that generates the face $F$. Then $F$ is also generated by $\nu:=\frac 1 {|F|} \sum_{j=1}^{|F|} \omega_j$, i.e.\ the normalized projective unit of $F$. This is because every face containing $\nu$ also contains all the $\omega_j$ (and vice versa), and $F$ is the smallest face with this property.

Now consider the state $\frac 1 2 \omega+\frac 1 2 \nu$. The smallest face that contains this state must be $G$. If this state had a decomposition into $|F|+1$ or fewer perfectly distinguishable pure states, then these would also generate $G$, and so $|G|\leq |F|+1$, in contradiction to our assumption. Thus any decomposition of $\frac 1 2 \omega+\frac 1 2\nu$ into perfectly distinguishable pure states uses at least $|F|+2$ states with non-zero coefficients, i.e.\ $S_0(\frac 1 2 \omega+\frac 1 2 \nu)\geq \log(|F|+2)$. But $\frac 1 2\omega+\frac 1 2 \nu=\frac 1 {2|F|} \sum_{j=1}^{|F|} \omega_j +\frac 1 2 \omega$ is a convex decomposition into $|F|+1$ pure states, thus
\[
   \check S_0\left(\frac 1 2 \omega+\frac 1 2 \nu\right)\leq\log(|F|+1)<\log(|F|+2)\leq S_0\left(\frac 1 2 \omega+\frac 1 2 \nu\right).
\]
It remains to show that $(v')\Rightarrow(ii')$. So suppose that (v') holds, but that there is a state $\omega\in\Omega_A$ with $\check S_0(\omega)<S_0(\omega)$; we will show that this leads to a contradiction. By definition of $\check S_0$, if this is the case, then there exist pure states $\omega_1,\ldots,\omega_n$ with $n=\exp(\check S_0(\omega))$ and $p_1,\ldots,p_n\geq 0$, $\sum_i p_i=1$, such that $\omega=\sum_{i=1}^n p_i\omega_i$. Using property (v'), and recursively looking at the faces generated by $\omega_1$, generated by $\omega_1,\omega_2$, generated by $\omega_1,\omega_2,\omega_3$ and so forth, shows that the rank of the face $G$ generated by $\omega_1,\ldots,\omega_n$ can be at most $n$. Since $\omega\in G$, this shows that $\omega$ can be decomposed into $n$ or fewer pure perfectly distinguishable states. Therefore $S_0(\omega)\leq\log n =\check S_0(\omega)$.\qed

\subsection{Some additional remarks on Theorem~\ref{ThmAllEquiv} and quantum logic}
\label{AppendixQuantumLogic}

The GPT framework has a close relation to quantum logic. This is not surprising, since much of the terminology of generalized probabilistic theories has appeared much earlier, in work on quantum logic and beyond. The approach via convex
  sets of states and observables can be traced back
  to Mackey~\cite{Mackey} (who immediately made connections to quantum
  logic), and was developed further through the 1960s and beyond.  A partial list
  of references includes~\cite{Mielnik, Gunson, LudwigAx}, the last
  two of which offer axiomatic characterizations of quantum theory.
  Interaction with the quantum logic tradition continued, with the
  orthomodularity of the lattice of faces of the state and/or effect
  spaces often providing a point of contact, especially in Ludwig's
  work~\cite{LudwigAx}.  Also closely related to the convex sets
  approach was the work of Foulis and Randall~\cite{Foulis-Randall66,
    Foulis-Randall74, Foulis-Randall81, Foulis-Randall81a} who, for
  example, studied ways of combining probabilistic systems.

Since Postulates 1 and 2 imply that the state cone $A_+$ is
self-dual (so coincides with the effect cone), and that each face of
this cone is the intersection of the cone with the image of a
\emph{filter} (equivalently given self-duality, a \emph{compresssion}
in the sense of \cite{ASBook}), we have from these postulates
alone, via e.g. \cite{ASBook}, Theorem 8.10, that the face
lattice is orthomodular.  The notion of orthomodular partially ordered
set, or its special case, the notion of orthomodular lattice, is often
taken to \emph{define} the notion of quantum logic, so we can say that
Postulates 1 and 2 imply that the face lattice of the cone of states,
(equivalently of the cone of effects, or of the set of normalized states) is a
quantum logic.

The covering property, in its most common variant the
\emph{atomic covering property}, states that for any element $x$ of the lattice
and any atom $a$ not below or equal to $x$, $a \join x$ covers $x$.\footnote{Here ``$y$ covers $x$'' means ``$y \ge x$, $y \ne x$, and
  there is no $z$ distinct from $x$ and $y$ such that $y \ge z \ge x$,
  i.e. ``$y$ is above $x$ with nothing in between.  An atom is an
  element that covers $0$.}  The equivalence of (iv) and (v), in a
setting more general than Postulates 1 and 2, is Proposition 9.7 of
\cite{ASBook} (first appearing in Proposition 4.2 in
\cite{AlfsenShultzJordan} and the discussion preceding it).  Along
with orthomodularity, the covering law was one of the assumptions of
Piron's famous lattice-theoretic characterization (\cite{Piron}; also
\cite{PironBook} and see the discussion in \cite{WilceLogic} or for
more detail and proofs, \cite{Varadarajan} pp.\ 18-38, 114-122) of a
class of lattices close to, although larger than, that of real,
complex, and quaternionic quantum theory.  A generalization of the
covering law was also used in Ludwig's axiomatization (see
e.g. \cite{LudwigAx}\footnote{The results in \cite{LudwigAx} were
    mostly obtained in a series of papers in the late 1960s and early
    1970s.}) of quantum theory within the convex sets framework, in
which the relevant lattice is a lattice of faces of the state space
(equivalent to a lattice of extremal effects in his context), and the
result characterized real, complex, and quaternionic quantum
theory.\footnote{Both Piron and Ludwig also made an atomicity
    assumption (which may be considered more technical than
    substantive, and always holds in finite dimension) and also
    assumed lattice dimension 4 or greater, so Hilbert spaces of
    dimension 3 or less were not dealt with, nor were spin factors or
    the exceptional Jordan algebra.  These low-dimensional cases also
    satisfy Piron's, and Ludwig's premises, but a theorem ruling out
    other instances satisfying them appears to be lacking.}

\end{document}